\documentclass[onecolumn,superscriptaddress, nofootinbib]{revtex4-2}
\usepackage{amsmath, amssymb,amsthm}
\usepackage{graphicx} 
\usepackage{xcolor}

\definecolor{royalblue(web)}{rgb}{0.25, 0.41, 0.88}

\newtheorem{definition}{Definition}

\newtheorem{theorem}{Theorem}
\newtheorem{corollary}{Corollary}[theorem]
\newtheorem{question}{Question}
\newtheorem{lemma}[theorem]{Lemma}

\usepackage{cleveref}

\newcommand{\ket}[1]{{|{#1} \rangle}}
\newcommand{\bra}[1]{{\langle {#1} |}}

\begin{document}

\title{Wavefunction Flows: Efficient Quantum Simulation of Continuous Flow Models}

\author{David Layden}
\affiliation{IBM Research \vspace{-1em}}
\affiliation{MIT-IBM Watson AI Lab \vspace{-1em}}

\author{Ryan Sweke}
\affiliation{African Institute for Mathematical Sciences (AIMS), South Africa \vspace{-1em}}
\affiliation{Department of Mathematical Sciences, Stellenbosch University, Stellenbosch 7600, South Africa\vspace{-1em}}
\affiliation{National Institute for Theoretical and Computational Sciences (NITheCS), South Africa\vspace{-1em}}

\author{Vojt\v{e}ch Havl\'{i}\v{c}ek}
\affiliation{IBM Research \vspace{-1em}}

\author{Anirban Chowdhury}
\affiliation{IBM Research \vspace{-1em}}

\author{Kirill Neklyudov}
\affiliation{Universit\'{e} de Montr\'{e}al\vspace{-1em}}
\affiliation{Mila - Quebec AI Institute\vspace{-1em}}
\affiliation{Institut Courtois}

\begin{abstract}
Flow models are a cornerstone of modern machine learning. They are generative models that progressively transform probability distributions according to learned dynamics. Specifically, they learn a continuous-time Markov process that efficiently maps samples from a simple source distribution into samples from a complex target distribution. We show that these models are naturally related to the Schr\"odinger equation, for an unusual Hamiltonian on continuous variables. Moreover, we prove that the dynamics generated by this Hamiltonian can be efficiently simulated on a quantum computer. Together, these results give a quantum algorithm for preparing coherent encodings (a.k.a., qsamples) for a vast family of probability distributions---namely, those expressible by flow models---by reducing the task to an existing classical learning problem, plus Hamiltonian simulation. For statistical problems defined by flow models, such as mean estimation and property testing, this enables the use of quantum algorithms tailored to qsamples, which may offer advantages over classical algorithms based only on samples from a flow model. More broadly, these results reveal a close connection between state-of-the-art machine learning models, such as flow matching and diffusion models, and one of the main expected capabilities of quantum computers: simulating quantum dynamics.
\end{abstract}

\maketitle

\section{Introduction}

Generative modeling is a fundamental task in machine learning (ML) which has seen enormous progress in the last decade. Its main goal is to model a target probability distribution $p_\text{targ}$ which is given empirically; i.e., given a dataset of samples generated from $p_\text{targ}$, one seeks to learn a model generating new samples from $p_\text{targ}$. This formulation, due to its generality, has revolutionized many fields of science including language modeling \citep{brown2020language}, computer vision \citep{ho:2020}, audio synthesis \citep{van2016wavenet}, protein folding \citep{jumper2021highly}, and modeling quantum states \cite{Carrasquilla_2019}.

Different approaches to this task have won out in different settings. We focus here on settings with continuous variables, where flow models achieve state-of-the-art performance~\cite{lipman2023flowmatchinggenerativemodeling, lipman2024flowmatchingguidecode}. These models learn a continuous-time Markov process that can be efficiently realized on a classical (i.e., non-quantum) computer, which transforms some source distribution~$p_0$ that we can easily sample from, into a distribution $p_T$ that approximates $p_\text{targ}$. Unlike some of their predecessors~\cite{yang2024diffusionmodelscomprehensivesurvey, ghojogh2022restrictedboltzmannmachinedeep, Kobyzev_2021}, their success owes partly to the fact that they describe dynamics---not thermal equilibrium---allowing them sidestep the issue of slow-mixing Markov chains that plagued earlier generative models, like Boltzmann machines \cite{ACKLEY1985}. 

In quantum computing, \emph{coherent encodings} of probability distributions, often known as \emph{qsamples} \cite{aharonov2003adiabaticquantumstategeneration,temme2025quantizedmarkovchaincouplings}, are an important class of states. For a probability distribution $p$, the corresponding qsample is the quantum state $\sum_x \sqrt{p(x)}\ket{x}$, if $x$ takes discrete values, or $\int \sqrt{p(x)} \ket{x} \, dx$ if $x$ takes continuous values. While measuring a qsample in the $\{\ket{x}\}$ basis simply gives a random sample from $p$, the ability to prepare qsamples of $p$ is more powerful than the ability to sample $x\sim p$. That is, if one can prepare qsamples for a distribution of interest, it is possible to construct quantum algorithms for a wide variety of distribution problems---such as mean estimation and property testing---which offer meaningful sample complexity advantages over optimal or best-known classical algorithms~\cite{Brassard_2002, brassard2011optimalquantumalgorithmapproximate, heinrich2001quantumsummationapplicationintegration, montanaro2015quantum, qchebyshev, qsubgaussian, knill2007optimal, Arunachalam2022simplerclassical, Cornelissen_2022, Bravyi_2011, gilyen2020distributional, chakraborty2010newresultsquantumproperty, li2018quantum, montanaro2013survey}. Qsample access also enables comparatively stronger models of learning compared to classical ones \cite{NEURIPS2023_acb7ce5a}, or improvements to distributional property testing \cite{gilyen2020distributional}. Accordingly, there has been long-standing interest in algorithms for generating qsamples on quantum computers \cite{kitaev2002classical, aharonov2003adiabaticquantumstategeneration, grover2002creatingsuperpositionscorrespondefficiently}, although this area remains much less developed than that of generative modeling in ML.

We propose a new quantum algorithm, at the intersection of these two ideas, that efficiently prepares qsamples by quantizing flow models. Specifically, we show that flow models naturally map to the Schrödinger equation, with an unusual Hamiltonian involving continuous variables. We call it the \textit{continuity Hamiltonian}, and refer to the quantum dynamics it generates as a \textit{wavefunction flow}. This continuity Hamiltonian is fully specified by the already-learned classical model, and requires no further learning, parameterization, or optimization. It evolves an initial fiducial state into a qsample for the distribution $p_T$ learned by the underlying flow model. And since flow models have been shown to efficiently generate a vast family of distributions, the corresponding Schrödinger equation can generate a correspondingly vast family of qsamples.

Crucially, flow models and the continuity Hamiltonian to which they map, are both defined on continuous variables. Turning this mapping into an algorithm for digital quantum computers poses a challenging Hamiltonian simulation problem, which comprises our main technical contribution. Specifically, we show how to discretize both space and time with efficiently controllable error, using Fourier collocation, a well-established numerical method. Besides handling the continuity Hamiltonian, which has an atypical form, our simulation technique gives a remarkably simple error bound which requires only mild regularity conditions, compared to earlier methods, thanks to our improved proof techniques. We then show that the state produced by this quantum simulation can be used for statistical inference problems defined by the classical flow model, for which quantum approaches may offer an advantage.

The resulting quantum algorithm has several potential applications in quantum computing (QC) and machine learning. But more fundamentally, it uses Hamiltonian simulation to form a new bridge between quantum complexity theory, in which qsamples are an important class of states, and classical machine learning. Concretely, it shows that the question of which qsamples are efficiently preparable can be cast as a theoretical ML question; and conversely, that the question of which distributions are efficiently representable with flow models can be cast as a quantum complexity question.

\section{Background}

\subsection{Qsamples}

The task of preparing qsamples was first studied by Aharonov and Ta-Shma \cite{aharonov2003adiabaticquantumstategeneration}, for the probability distribution defined by a classical circuit evaluated on uniformly random inputs. They gave efficient polynomial-time quantum algorithms to prepare qsamples corresponding to limiting distributions of certain rapidly mixing Markov chains. Qsample generation has since been studied extensively in the context of quantizing classical Markov chain Monte Carlo methods, where quantum algorithms are known to give polynomial speed-ups over classical algorithms \cite{somma2008annealing,harrow2010adaptive,Arunachalam2022simplerclassical,wocjan2008quantumsampling,montanaro2015quantum,Orsucci2018fasterquantummixing,lemieux2024quantumsamplingalgorithmsquantum,temme2025quantizedmarkovchaincouplings}. However, these are of limited efficacy if the underlying Markov chain is slow-mixing, as is often the case in practice \cite{goodfellow_book}.

Efficient quantum algorithms for qsample preparation are known when $p$ is simple to specify, e.g., Gaussian or efficiently integrable, and a number of authors have undertaken detailed resource estimation in such cases \cite{grover2002creatingsuperpositionscorrespondefficiently, kitaev2009wavefunctionpreparationresamplingusing, mcardle2025quantumstatepreparationcoherent, rattew2022preparingarbitrarycontinuousfunctions, lemieux2024quantumsamplingalgorithmsquantum}. There are also a number of heuristic proposals whose performance is more difficult to quantify \cite{liu:2018, dallaire:2018, lloyd:2018, amin:2018, zoufal:2019, coyle2020born, wild2021quantumsampling, coopmans:2024, zhang:2024}. However, the problem of preparing qsamples in general is likely hard: it was shown that efficient qsample preparation for all distributions that can be efficiently sampled classically would lead to unlikely complexity theoretic conclusions \citep{aharonov2003adiabaticquantumstategeneration}.

\subsection{Flow models}\label{sec:flow-models}

Consider the first-order ordinary differential equation (ODE)
\begin{equation}
\frac{d}{dt} x_t = v_t(x_t),
\label{eq:ODE}
\end{equation}
where $x_t \in \mathbb{R}^d$ for time $t \in [0,T]$, and where $v: [0,T] \times \mathbb{R}^d \rightarrow \mathbb{R}^d$ is a time-dependent vector field that is generally non-linear. This equation can be interpreted as describing the position of a particle moving through $d$-dimensional space with time- and position-dependent velocity $v_t$, so the latter is often called the \textit{velocity field}. For a given initial condition $x_0$, the mapping $x_0 \mapsto x_t$ defined by Eq.~\eqref{eq:ODE} is a flow, which is fully specified by the velocity field $v_t$.

The ODE in \eqref{eq:ODE} is fully deterministic. However, if we sample a random initial condition $x_0 \sim p_0$ from some distribution $p_0$, then the time-evolved vector $x_t \sim p_t$ will also be a random variable with some probability density function (PDF) $p_t$, for all $t \in [0,T]$. This defines a continuous-time Markov process, and the family of marginal distributions $(p_t)_{t \in [0,T]}$ is called the \textit{probability path} generated by $v_t$. By invoking the change of variable formula for PDFs on infinitesimal timestep of Eq.~\eqref{eq:ODE}, one can show that $p_t$ satisfies the \textit{continuity equation}
\begin{equation}
\frac{\partial} {\partial t} p_t(x)
=
- \nabla \cdot \Big[ v_t(x) \, p_t(x) \Big]
\label{eq:continuity_eq}
\end{equation}
with respect to the velocity field $v_t$, for time $t\in [0,T]$. 

Notice that while the ODE in Eq.~\eqref{eq:ODE} and the continuity equation \eqref{eq:continuity_eq} are closely related, they describe the evolution of very different objects. For a given initial condition, Eq.~\eqref{eq:ODE} describes the dynamics of a $d$-dimensional vector, and can be solved numerically using standard ODE solvers, provided $v_t$ is well-behaved and $d$ is not unreasonably large. On the other hand, the PDF $p_t$ assigns a non-negative density $p_t(x)$ to every $x \in \mathbb{R}^d$, so numerically solving the continuity equation involves tracking $p_t(x)$ for all $x$ and $t$, which would require vastly more computational resources. There is no need to do so, however, because the ODE in \eqref{eq:ODE}, with a random initial condition $x_0 \sim p_0$, simulates the continuity equation \eqref{eq:continuity_eq}, in that solving \eqref{eq:ODE} lets us efficiently sample from any $p_t$ in \eqref{eq:continuity_eq} without needing to compute the density $p_t(x)$ explicitly. This is reminiscent of the sense in which a quantum computer can simulate the Schr\"odinger equation without explicitly computing any quantum amplitudes.

Flow models aim to learn a velocity field $v_t$ under which the continuity equation transforms a simple source distribution $p_0$ at $t=0$ into $p_\text{targ}$ at $t=T$. Many probability paths connect these two distributions, and accordingly, many velocity fields are possible.  In practice, a chosen ground truth velocity field $v_t$ is approximated by some parameterized function $v_t^\theta$, usually implemented by a neural network with weights $\theta$. The learned velocity field $v_t^\theta \approx v_t$ then transforms $p_0$ into a distribution $p_T$ approximating $p_\text{targ}$. One can easily sample from this $p_T$ by sampling $x_0 \sim p_0$, and using it as an initial condition to numerically solve the ODE \eqref{eq:ODE}, with $v_t^\theta$ in place of $v_t$. In the simplest case, $p_0$ is taken to be a $d$-dimensional isotropic Gaussian (i.e., with covariance matrix $\propto \! I$), and $p_\text{targ}$ is a distribution from which we are given training samples. While we will mostly focus on this simple case, flow models easily handle more complex sampling problems too, such as conditional sampling.

There are several principled ways to learn a velocity field that transforms $p_0$ into $p_T \approx p_\text{targ}$. For our quantum algorithm, we only care that such a velocity field has been learned, but we are agnostic as to how it was learned. To provide some context, however, we briefly discuss two prominent methods of learning $v_t^\theta$.

\subsubsection{Flow matching}
\label{sec:flow_matching}

Flow models, as presented above, were first introduced under the name \textit{continuous normalizing flows\footnote{The word ``continuous'' here differentiates this technique from earlier ones called \textit{normalizing flows}, which use a finite number of discrete steps instead \cite{papamakarios:2021}.}}, and were trained using a maximum likelihood objective, or equivalently, by minimizing an unbiased estimate of the KL-divergence between $p_T$ and $p_\text{targ}$ \cite{chen2019}. However, this training procedure required repeatedly solving the ODE \eqref{eq:ODE} for different candidate velocity fields, and was therefore impractically slow. Recently, simulation-free training of continuous normalizing flows---the framework learning $v_t^\theta$ without the computational bottleneck of ODE integration---has been developed independently by different groups from different perspectives: \textit{flow matching} \cite{lipman:2023}, \textit{rectified flows} \cite{liu2022flow}, \textit{stochastic interpolants} \cite{albergo2023stochastic}, and \textit{action matching} \cite{neklyudov2023action}. Here, we briefly summarize \textit{flow matching}; the most popular, minimalistic approach which currently offers state-of-the-art performance.

We denote the PDF of a $d$-dimensional isotropic Gaussian distribution as
\begin{equation}
\mathcal{N}(x \, | \,\mu, \sigma^2 I) 
=
(2 \pi \sigma^2)^{-d/2} \exp \left( - \frac{\|x-\mu\|^2}{2 \sigma^2} \right),
\end{equation}
where $\| \cdot \|$ is the Euclidean norm, $\mu \in \mathbb{R}^d$ is the mean, and $\sigma^2 I$ is the covariance for some $\sigma>0$. The simplest version of flow matching uses the source distribution $p_0 = \mathcal{N}( \, \cdot \,   |\,  0, \,  I)$, and by convention, sets $T=1$ (or equivalently, absorbs $T$ into the velocity field). Then, for any $t \in [0,1]$, it defines the random variable 
\begin{equation}
x_t = (1-t) x_0 + t x_\text{targ} \; \sim p_t,
\end{equation}
where $x_0 \sim p_0$ and $x_\text{targ} \sim p_\text{targ}$ are sampled independently from the source and target distributions, respectively. This specifies a probability path. It is simple to show that the PDF of this $x_t$ is a convolution of $p_\text{targ}$ with an appropriately scaled Gaussian, namely:
\begin{equation}
p_t(x) = \int_{\mathbb{R}^d}   \mathcal{N} \Big[ x \, \big|\,  t x', (1-t)^2 I \Big] \, p_\text{targ}(x') dx',
\label{eq:linear_path_marginal}
\end{equation}
which reduces to $p_0$ at $t=0$ and to $p_\text{targ}$ as $t \rightarrow 1$. 

Under mild regularity conditions (see, e.g., \citep{ambrosio2005gradient}), there exists a velocity field $v_t$ such that the density $p_t$ from Eq.~\eqref{eq:linear_path_marginal} satisfies the continuity equation \eqref{eq:continuity_eq}. In principle, if this $v_t$ were known, one could model it with a parametric ansatz $v_t^\theta$ by minimizing the following \textit{flow matching loss} with respect to the parameters $\theta$:
\begin{equation}
\mathcal{L}_{\text{FM}}(\theta)
=
\mathbb{E}_{t \sim \text{unif}[0,1]} \; 
\mathbb{E}_{x_0 \sim p_0} \; 
\mathbb{E}_{x_\text{targ} \sim p_\text{targ}}
\Big \| v_t^\theta \big[ (1-t) x_0 + t x_\text{targ} \big] - v_t \big[(1-t) x_0 + t x_\text{targ}\big]
\Big \|^2\,.
\end{equation}
In other words, one could generate the samples from $p_t$ as the linear interpolation of samples drawn from $p_0$ and $p_\text{targ}$, and then minimize the quadratic error between the modeled velocity field and the ground truth. Of course, in practice, the ground truth velocity field $v_t$ that generates $p_t$ is unknown, which makes this objective impossible to optimize.

The key step that enables efficient training of generative models is the derivation of the \textit{conditional flow matching loss} $\mathcal{L}_\text{CMF}(\theta)$, which is equivalent to $\mathcal{L}_{\text{FM}}(\theta)$, but amenable to optimization. In particular, one can show that
\begin{align}
\nabla_\theta \mathcal{L}_{\text{FM}}(\theta) 
= 
\nabla_\theta ~&\mathcal{L}_{\text{CFM}}(\theta)\,,\\
\text{where}\;\;~&\mathcal{L}_{\text{CFM}}(\theta)
=
\mathbb{E}_{t \sim \text{unif}[0,1]} \; 
\mathbb{E}_{x_0 \sim p_0} \; 
\mathbb{E}_{x_\text{targ} \sim p_\text{targ}}
\Big \| v_t^\theta \big[ (1-t) x_0 + t x_\text{targ} \big] - (x_\text{targ} - x_0)
\Big \|^2\,.
\end{align}
Note that the new objective $\mathcal{L}_{\text{CFM}}(\theta)$ does not depend on the ground truth velocity field $v_t$. Moreover, it allows for an efficient and unbiased Monte Carlo estimate by sampling $x_\text{targ} \sim p_\text{targ}$ from training data, $t \sim \text{unif}[0,1]$, and $x_0 \sim p_0$---all of which is easy---and evaluating the parametric ansatz at their linear interpolation. Crucially, it does not require numerically solving a differential equation. In principle, if we could minimize $\mathcal{L}_{\text{CFM}}$ exactly, we would recover $v_t$ and could therefore generate new samples from $p_\text{targ}$. In practice, the learned parametric ansatz $v_t^\theta \approx v_t$ is used instead of $v_t$ in the ODE \eqref{eq:ODE} that, when solved with a random initial condition $x_0 \sim p_0$, produces a sample $x_T$ from a distribution $p_T \approx p_\text{targ}$.

\subsubsection{Diffusion models as flow models}
\label{sec:diffusion}

Another common way to learn a velocity field $v_t^\theta$ comes from diffusion models \cite{Sohl-Dickstein:2015, ho:2020, song:2021}. Diffusion models are a different type of generative model that also learn continuous-time Markov processes, and that predate flow matching. While they are not exactly flow models, they can easily be mapped to such. 

Specifically, diffusion models are naturally defined in terms of two stochastic differential equations (SDEs). Consider first the \textit{forward SDE} for a random vector $x_t \in \mathbb{R}^d$:
\begin{equation}
dx_t = f_t(x_t) \, dt + g_t \, dw_t
\end{equation}
for time running from $t=0$ to $t=T$, where $f_t$ and $g_t$ are well-behaved functions called drift and diffusion coefficients respectively, that will be specified later, and $w_t$ is a $d$-dimensional Wiener process. We denote the marginal PDF of the state $x_t$ at time $t$ as $p_t$, where $x_0 \sim p_0$ is the the initial condition. Next, consider the \textit{reverse SDE} for a random vector $y_t \in \mathbb{R}^d$:
\begin{equation}
dy_t = \Big[ f_t(y_t) - g_t^2 \, \nabla \ln p_t(y_t) \Big] dt + g_t \, d\bar{w}_t
\label{eq:reverse_SDE}
\end{equation}
with time running backwards from $t=T$ to $t=0$, where $\bar{w}_t$ is a $d$-dimensional Wiener process with time flowing backwards, and $dt$ denotes an infinitesimal negative timestep. If the initial condition for the reverse SDE is $y_T \sim p_T$, then one can show---with considerable mathematical effort---that $y_t \sim p_t$ \cite{ANDERSON1982313}. That is, the time-evolved states $x_t$ and $y_t$, from the forward and reverse SDEs respectively, are identically distributed random variables, for all $t \in [0,T]$. 

Diffusion models use the forward SDE for learning, and the reverse SDE for generating new samples. Specifically, they use $x_0 \sim p_0 = p_\text{targ}$ as the initial condition\footnote{Unfortunately, the literature on diffusion models associates $t=0$ with $p_\text{targ}$ and $t=T$ with a source distribution, which is opposite to the convention used in later flow models, and introduced above.} in the forward equation. They then pick drift and diffusion coefficients $f_t$ and $g_t$ so that $p_t \rightarrow \mathcal{N}( \, \cdot \, |\, 0, \, I)$ quickly as $t$ grows, regardless of $p_\text{targ}$, so that $p_T \approx \mathcal{N}( \, \cdot \, |\, 0, \, I)$ for sufficiently large $T$. A convenient choice, implicitly used in denoising diffusion probabilistic models (DDPM) \cite{ho:2020, song:2021}, is 
\begin{equation}
f_t(x_t) = -x_t \beta_t/2 \quad \text{and} \quad g_t = \sqrt{\beta_t},
\label{eq:DDPM_functions}
\end{equation}
where $\beta_t$ is some affine, increasing function of $t$. Then, the forward SDE is repeatedly solved, to generate samples $x_t \sim p_t$, which are used to estimate the function 
\begin{equation}
s_t(x) = \nabla \ln p_t(x)
\end{equation}
called the \textit{score} of $p_t$, through a well-established statistical process called \textit{score matching} \cite{implicit_SM, conditional_SM}. Perhaps surprisingly, this is much easier than estimating the density function $p_\text{targ}(x)$ directly \cite{NEURIPS2019_3001ef25}. We denote the learned score function as $s_t^\theta(x)$. To generate new samples from a distribution close to $p_\text{targ}$, it suffices to sample $y_T \sim \mathcal{N}( \, \cdot \, |\, 0, \, I)$ and solve the reverse SDE \eqref{eq:reverse_SDE} with $s_t^\theta \approx s_t$ in place of $\nabla \ln p_t$. 

Finally, it is possible to show that the ODE
\begin{equation}
\frac{d}{dt} x_t = f_t(x_t) - \frac{1}{2} g_t^2 \, \nabla  \ln p_t(x_t),
\label{eq:probability_flow_ODE}
\end{equation}
with random initial condition $x_T \sim p_T$, produces a random time-evolved state $x_t \sim p_t$ when solved backwards in time from $t=T$. So once the score function has been learned, one can equivalently generate new samples by solving this ODE---which defines a flow model---rather than the reverse SDE, with $s_t^\theta$ again in place of $\nabla \ln p_t$. Eq.~\eqref{eq:probability_flow_ODE} is known as the \textit{probability flow ODE} \cite{song:2021, DDIM}, although it can be viewed simply as an instance of Eq.~\eqref{eq:ODE} with a particular choice of velocity field expressed in terms of the score function. Note that when $f_t$ is a conservative vector field (i.e., the gradient of some scalar-valued function), as in Eq.~\eqref{eq:DDPM_functions}, the right-hand side of the probability flow ODE~\eqref{eq:probability_flow_ODE} is also conservative.

\section{Mapping the continuity equation to the Schr\"odinger equation}
\label{sec:flow_quantum_mapping}

We show in this section that flow models are closely connected to quantum dynamics. 

\subsection{General Case}

Suppose a probability path $(p_t)_{t\in [0,T]}$ satisfies the continuity equation \eqref{eq:continuity_eq} with respect to a velocity field $v_t$, in $d$-dimensional space. Then it is simple to show that the corresponding qsample, namely, the wavefunction $\Psi_t(x) = \sqrt{p_t(x)}$, satisfies the Schr\"odinger equation:
\begin{align}
i \frac{\partial }{\partial t} \Psi_t(x)
&=  \; \; 
\frac{i}{2 \sqrt{p_t(x)}} \frac{\partial}{\partial t} p_t(x) 
\; \; = \; \;  
-\frac{i}{2 \Psi_t(x)} \; \nabla \cdot \Big[ v_t(x) \Psi_t(x)^2 \Big] \nonumber \\[1ex]
&= 
-\frac{i}{2 \Psi_t(x)} \bigg \{ \Psi_t (x) \; \nabla \cdot \Big[ v_t(x) \Psi_t(x) \Big]
+
\Big[ v_t(x) \Psi_t(x) \Big] \cdot \nabla \Psi_t(x) \bigg \} \label{eq:schrodinger_derivation} \\[1ex]
&= \qquad \qquad \,
\underbrace{
\frac{1}{2} \bigg \{  
(-i \nabla) \cdot \Big[ v_t(x) \Psi_t(x) \Big]
+ v_t(x) \cdot (-i \nabla)\Psi_t(x) \bigg \}}_{\hat{\mathcal{H}}_t \Psi_t(x)} . \nonumber
\end{align}
That is, if $p_t$ obeys the continuity equation with velocity field $v_t$, then $\Psi_t = \sqrt{p_t}$ obeys the Schr\"odinger equation $i \frac{\partial}{\partial t} \Psi_t = \hat{\mathcal{H}}_t \Psi_t$ with Hamiltonian 
\begin{equation}
\hat{\mathcal{H}}_t = \frac{1}{2} 
\Big[ 
\hat{p} \cdot v_t(\hat{x})
+
v_t(\hat{x}) \cdot \hat{p}
\Big],
\label{eq:continuity_Hamiltonian}
\end{equation}
where $\hat{x}=(\hat{x}_1, \dots, \hat{x}_d)$ and $\hat{p} = (\hat{p}_1, \dots, \hat{p}_d)$ are $d$-dimensional position and momentum operators, whose components act as $(\hat{x}_j \Psi_t)(x) = x_j \Psi_t(x)$ and $(\hat{p}_j \Psi_t)(x) = -i \frac{\partial}{\partial x_j} \Psi_t(x)$, respectively, at $x = (x_1, \dots, x_d)$. We therefore refer to $\hat{\mathcal{H}}_t$ from Eq.~\eqref{eq:continuity_Hamiltonian} as the \textit{continuity Hamiltonian} in this context, and the quantum dynamics it generates as a \textit{wavefunction flow} (in analogy to the probability flow of \cite{song:2021}). While the derivation here tacitly assumes $p_t(x) >0$ everywhere, we show in Appendix~\ref{app:continuity_H} that this assumption is not necessary. 

This close correspondence between the continuity equation---which underlies flow models---and the Schr\"odinger equation with Hamiltonian $\hat{\mathcal{H}}_t$, reveals a strikingly natural connection between modern machine learning and the dynamics of quantum systems. In particular, it points to a simple quantum algorithm to prepare a broad family of qsamples, allowing one to probe the learned distribution $p_T$ quantum mechanically.
\begin{description}
\item[1 (Already done in flow models)] Specify a probability path connecting a simple source distribution $p_0$ to a target distribution $p_\text{targ}$, which implicitly defines a velocity field $v_t$ that generates this path. Learn to approximate the ground truth velocity field $v_t$ with a parameterized function $v_t^\theta$, that generates a nearby probability path  $(p_t)_{t\in[0,T]}$ where $p_T \approx p_\text{targ}$. Most commonly, this is done using training samples from $p_\text{targ}$, e.g., as explained in Sections~\ref{sec:flow_matching} and \ref{sec:diffusion}.
\item[2] Prepare the initial qsample $\Psi_0 = \sqrt{p_0}$ on a quantum computer. For instance, if $p_0 = \mathcal{N}(\, \cdot \, | \, 0, \, I)$, then $\Psi_0$ is simply a tensor product of 1-dimensional standard normal qsamples, which are easy to prepare \cite{grover2002creatingsuperpositionscorrespondefficiently, kitaev2009wavefunctionpreparationresamplingusing}. 
\item[3] Rather than query the learned velocity field $v_t^\theta$ with an ODE solver (to solve Eq.~\eqref{eq:ODE}) on a classical computer, query it on a quantum computer to do Hamiltonian simulation. That is, simulate evolution by the continuity Hamiltonian $\hat{\mathcal{H}}_t$, with $v_t^\theta$ in place of $v_t$, for time $t\in [0,T]$.
\end{description}
In principle, these steps provably produce the wavefunction $\Psi_T(x) = \sqrt{p_T(x)}$. The distance between this state and the target qsample $\sqrt{p_\text{targ}(x)}$ depends solely on the distance between the distributions $p_T$ and $p_\text{targ}$, which can be analyzed through purely classical means \cite{chen2023sampling, benton2024error}. In practice, realizing steps 2 and 3 on a digital quantum computer will also incur some discretization error. We show in Section~\ref{sec:qsim-algorithm} that this additional error can efficiently be made arbitrarily small.

\subsection{Special case: conservative velocity fields}

Without loss of generality, one can restrict the family of the velocity fields used by the flow models to conservative fields, i.e., $v_t(x) = \nabla V_t(x)$ for some scalar-valued potential function $V_t$. This is due to a celebrated result for probability measures: Theorem 8.3.1 in \citep{ambrosio2005gradient}. In particular, for every probability path $(p_t)_{t\in [0,T]}$ that represents an absolutely continuous curve in the 2-Wasserstein space, one can define a unique potential function $V_t$ such that the corresponding conservative velocity field satisfies the continuity equation \eqref{eq:continuity_eq}. In Sec.~\ref{sec:diffusion}, we demonstrate the explicit form of this potential for diffusion, which allows for its efficient parameterization \citep{salimans2021should}. Furthermore, one can efficiently learn the corresponding potential function for any type of a process while having access only to uncorrelated samples from the marginal densities \citep{neklyudov2023action}. 

When $v_t$ is conservative, the continuity Hamiltonian $\hat{\mathcal{H}}_t$ can be expressed in a simpler form:
\begin{equation}
\hat{\mathcal{H}}_t^c = i \big [ \hat{K}, \, V_t(\hat{x}) \big],
\label{eq:conservative_continuity_Hamiltonian}
\end{equation}
where $[\, \cdot \, , \cdot \,]$ denotes a commutator, and $\hat{K} = \frac{1}{2} \, \hat{p} \cdot \hat{p}$ can be interpreted as the kinetic energy operator for a particle of unit mass, which acts as $(\hat{K} \Psi_t)(x) = -\frac{1}{2} \nabla^2 \Psi_t(x)$. This can be seen by substituting $v_t = \nabla V_t$ into the last line of Eq.~\eqref{eq:schrodinger_derivation}, and using the vector calculus identity $\nabla^2 V_t \Psi_t  = \Psi_t \nabla^2 V_t + 2 (\nabla V_t) \cdot (\nabla \Psi_t) + V_t \nabla^2 \Psi_t$ to get
\begin{align}
\hat{\mathcal{H}}_t \Psi_t(x)
&= 
-\frac{i}{2} \Big[ 
\nabla \cdot \big[\Psi_t(x) \nabla V_t(x) \big]
+ (\nabla V_t)(x) \cdot (\nabla \Psi_t)(x) 
\Big]  \nonumber \\
&= 
-\frac{i}{2} \Big[ 
\Psi_t(x) \nabla^2 V_t(x) 
+ 2 (\nabla V_t)(x) \cdot (\nabla \Psi_t)(x) \Big] \\ 
&= 
- \frac{i}{2} \Big\{ 
\nabla^2 \big[V_t(x) \Psi_t(x)\big] - V_t(x) \nabla^2 \Psi_t(x)
\Big \}
\qquad \qquad \quad = \;\; 
\hat{\mathcal{H}}_t^c  \, \Psi_t(x). \nonumber
\end{align}

\subsection{Context}

Before analyzing the cost of simulating the continuity Hamiltonian on a digital quantum computer in the next section, we pause to comment on some more fundamental aspects of the results above.

First, a different connection between the Schr\"odinger equation and continuity equation has been known since the early days of quantum mechanics: the Madelung equations \cite{Madelung1927}. Suppose a wavefunction $\Phi_t(x)$ evolves by the  Schr\"odinger equation with the more standard Hamiltonian $\hat{K} + V_t(\hat{x})$. The Madelung equations re-express its dynamics using two real differential equations: a continuity equation for the probability density $|\Phi_t(x)|^2$, and another equation for the phase $\arg[ \Phi_t(x)]$. Crucially, these differential equations are coupled, so the effective velocity field in the first depends on phase, which in turn evolves according to the probability density. What is different about the continuity Hamiltonian $\hat{\mathcal{H}}_t$ defined above is that---by construction---it decouples these two equations, so that the phase stays fixed, while the velocity field can be specified directly. Of course, $\hat{\mathcal{H}}_t$ may therefore not occur in nature, as it does not have the canonical form of kinetic plus potential energy. But it can nevertheless be efficiently simulated on a quantum computer, as we show in the next section.

Second, note that while $V_t(\hat{x})$ in Eq.~\eqref{eq:conservative_continuity_Hamiltonian} can be viewed as a potential energy operator (in that it acts through pointwise multiplication), the PDF $p_t$ is unrelated to the Boltzmann distribution defined by $V_t$. That is because neither flow models, nor their corresponding wavefunction flows, implement a convergent process that asymptotes to some steady state. Crucially, this means that neither suffers from slow thermalization/mixing. Rather, both are fundamentally dynamical. In particular, our quantum algorithm is more reminiscent Grover's algorithm \cite{grover:1996}, in that it rotates a simple quantum state into a complicated one in some finite time $T$, and would overshoot if it were somehow run for longer.

\section{Digital quantum simulation of the continuity Hamiltonian}\label{sec:qsim-algorithm}

In this section, we analyze the space and time complexity of simulating the continuity Hamiltonian on a digital quantum computer. We focus on the important case where the velocity field is conservative ($v_t = \nabla V_t$), although we expect the general case can be handled similarly. Concretely, we consider a probability path $(p_t)_{t\in[0,T]}$ generated by $\nabla V_t$ for a given potential potential function $V_t$. Rather than produce the wavefunction $\sqrt{p_T(x)}$ corresponding to the PDF $p_T$, we aim to prepare the finite-dimensional state
\begin{equation}
\ket{\psi_T} \propto \sum_x \sqrt{p_T(x)} \ket{x}
\label{eq:discrete_qsample}
\end{equation}
to within a specified error $\epsilon$, where the sum is over an appropriate grid of points in $\mathbb{R}^d$. We will do so using only standard quantum operations and quantum oracle access to $V_t$, rather than to potentially complicated integrals thereof (which are not learned by flow models), as in many first-quantized simulation methods based on Galerkin discretization \cite{su:2021}.

Digitally simulating the dynamics generated by $\hat{\mathcal{H}}_t^c$ involves three main challenges: first, we must discretize space and bound the resulting error. Second, the ensuing finite-dimensional Hamiltonian depends on time, so we must bound the error from approximating it with a piecewise-continuous (in $t$) Hamiltonian. And third, this last Hamiltonian involves commutators of simple terms, rather than sums thereof, so several well-known techniques (e.g., moving to the interaction picture of the discretized kinetic energy operator) do not apply directly. The first point, presented in Sec.~\ref{sec:discretizing_space}, is our main technical contribution of this section, since $\hat{\mathcal{H}}_t^c$ does not have the widely-analyzed form of kinetic plus potential energy. For the second and third points, presented in Sec.~\ref{sec:discretizing_time}, we have adapted existing techniques to bound the time complexity.

We begin by summarizing the results. Considering a $d$-dimensional cube of edge length $L$, we show how prepare the discrete qsample $\ket{\psi_T}$ from Eq.~\eqref{eq:discrete_qsample} to within an error tolerance $\epsilon \in (0, 2T]$ using a simple product formula with
\begin{equation}
r=O \left( \frac{d^{\,6} \, T^{2+2/s}}{\epsilon^{1+2/s}} \right)
\end{equation}
timesteps and 
\begin{equation}
n = O
\bigg( d \log \Big[ Ld \, \big ( T/\epsilon \big)^{1/s} \Big] \bigg)
\end{equation}
qubits, plus additional ancillas to compute/uncompute $V_t$, where $s = \Omega(d)$ is an adjustable parameter related to the smoothness of $p_t$ and $V_t$. (Larger $s$ gives better scaling, but requires smoother functions and could yield larger multiplicative factors.) We require only mild regularity conditions concerning the smoothness of these functions, which are often satisfied by construction in flow models. The big-$O$ notation here suppresses multiplicative factors that control the norms/derivatives of $p_t$ and $V_t$, which are unavoidable in the analysis of such differential equations. Exact lower bounds on $r$ and $n$ are given in Eqs.~\eqref{eq:dlogN_bound} and \eqref{eq:r_bound} respectively. 

\subsection{Discretizing space}
\label{sec:discretizing_space}

We will simulate the dynamics of $\hat{\mathcal{H}}_t^c$ on a $d$-dimensional torus $\mathbb{T}^d$, defined as the cube $[0,L]^d$ with opposite sides identified, for some edge length $L>0$. This simulation domain has two important features: it is bounded, and it has periodic boundaries. The first feature is already present (implicitly) in classical flow models, to avoid overflow, and its impact can be made negligible by choosing a sufficiently large $L$. The second feature seems more uniquely quantum, and ensures that the Laplace operator $\nabla^2$ can be discretized in the same way everywhere, so we do not have to treat a boundary separately. As with the first feature, we can handle distributions that are not inherently periodic by choosing $L$ to be large enough that most of their probability mass is contained in $[0,L]^d$, then approximating them by their corresponding wrapped distribution.

To discretize $\mathbb{T}^d$, we first define the 1-dimensional grid
\begin{equation}
\mathbb{X}_N = \frac{L}{N} \, \Big \{ 0, 1, \dots, N-1 \Big\} \; \subset \mathbb{T}^1,
\label{eq:X_def}
\end{equation}

where $N\ge 1$ is some integer power of $2$ that we will choose in order to control the discretization error. Then the $d$-dimensional grid $\mathbb{X}_N^d \subset \mathbb{T}^d$ comprises $N^d$ points covering $\mathbb{T}^d$ uniformly, with spacing $L/N$ along each spatial dimension. We will seek to prepare the $n=d\log_2(N)
$ qubit state
\begin{equation}
\ket{\psi_t} \propto \sum_{x\in \mathbb{X}_N^d} \sqrt{p_t(x)} \ket{x}
\; \in \mathbb{C}^{N^d},
\label{eq:Psi_t}
\end{equation}
for $t=T$, with $\|\ket{\psi_t}\|=1$ under the Euclidean norm. Since the grid points $x=(x_1, \dots, x_d) \in \mathbb{X}_N^d$ generally have non-integer coordinates, we use the shorthand notation
\begin{equation}
\ket{x} = \ket{ x_1  \,N / L}  \otimes \ket{ x_2  \,N / L} \otimes \cdots \otimes \ket{ x_d  \,N / L},
\label{eq:|x>_notation}
\end{equation}
with integers $x_j N/L \in \{0, \dots, N-1\}$ indexing one-hot vectors $\ket{ x_j  \,N / L} \in \mathbb{C}^N$ in the usual way.

In this section, we define a finite-dimensional Hamiltonian whose dynamics generate a state arbitrarily close to $\ket{\psi_t}$. We used a Fourier pseudo-spectral discretization \cite{lubich2008quantum}, inspired by Ref.~\cite{Childs2022}, to construct it, and derive a bound on the resulting spatial discretization error. Here, however, we simply define said Hamiltonian and state the error bound, with little reference to the underlying formalism, and relegate the lengthy proofs to Appendix~\ref{app:discretizing_space}. 

To prepare $\ket{\psi_t}$, we define the finite-dimensional \textit{discrete continuity Hamiltonian}: 
\begin{equation}
H_t^c = i[K, D_{V_t}],
\label{eq:discrete_H}
\end{equation}
where the $N^d \times N^d$-dimensional matrices $K$ and $D_{V_t}$ are pseudospectral discretizations of the kinetic and potential energy operators, $\hat{K}$ and $V_t(\hat{x})$ respectively, from the conservative form of the continuity Hamiltonian, $\hat{\mathcal{H}}_t^c$, in Eq.~\eqref{eq:conservative_continuity_Hamiltonian}. Specifically, we define the diagonal matrix
\begin{equation}
D_{V_t} = \sum_{x \in \mathbb{X}_N^d} V_t(x) \ket{x} \! \bra{x}
\label{eq:D_V_t}
\end{equation}
which enacts pointwise multiplication by $V_t$ on grid points. Then to construct $K$, we define an $N \times N$ matrix that is a discretization of $-\frac{1}{2} \frac{\partial^2}{\partial x_j^2}$ on the $\log_2(N)$-qubit register describing the $j^\text{th}$ spatial coordinate, and combine $d$ copies of it using a Kronecker sum. To that end, let 
\begin{equation}
F = \frac{1}{\sqrt{N}} \sum_{j,\ell=0}^{N-1} \exp \left( \frac{i 2\pi j \ell}{N} \right) \ket{j}\!\bra{\ell}
\label{eq:F}
\end{equation}
be the $N\times N$ quantum Fourier transform matrix, and define the diagonal matrices
\begin{equation}
S = \sum_{j=0}^{N-1} (-1)^j \ket{j} \! \bra{j}
\label{eq:S}
\end{equation}
and
\begin{equation}
D_K = \left( \frac{2 \pi}{L} \right)^2 \; \sum_{j=0}^{N-1} \left(j -\frac{N}{2} \right)^2 \ket{j} \! \bra{j}
\label{eq:D_K}
\end{equation}
of the same size, following Refs.~\cite{Childs2021highprecision, Childs2022}. Note that $F$ and $S$ are both unitary, while $D_K$ is Hermitian. Finally, define
\begin{equation}
K = \frac{1}{2} \Big( S F D_K F^\dag S^\dag \Big)^{\oplus \, d},
\label{eq:K_def}
\end{equation}
where $M^{\oplus d} = (M \otimes I^{\otimes d-1}) + (I \otimes M \otimes I^{\otimes d-2}) + \cdots + (I^{\otimes d-1} \otimes M)$ denotes the $d$-fold Kronecker sum of a matrix $M$. One way to interpret $K$ is as a finite-difference stencil that involves not just the nearest or next-nearest neighbors of a grid point $x \in \mathbb{X}_N^d$, but rather, all grid points to appropriate degrees. Crucially, it can easily be diagonalized using quantum Fourier transforms.

Now denote as $\ket{\phi_t} \in \mathbb{C}^{N^d}$ the solution to the finite-dimensional Schr\"odinger equation with Hamiltonian $H_t^c$,
\begin{equation}
i \frac{d}{dt} \ket{\phi_t} = H_t^c \ket{\phi_t},
\end{equation}
and with an initial condition $\ket{\phi_0}$ meant to approximate the ideal initial state $\ket{\psi_0}$ defined by Eq.~\eqref{eq:Psi_t}. Then $\ket{\phi_T}$ can be made arbitrarily close to the target state $\ket{\psi_T}$ by choosing a sufficiently large $N$, as described by Theorem~\ref{thm:spatial_disc_error}. To state this error bound, we must first introduce some definitions.

\textbf{Definitions.} Define $\|f\|_{\mathsf{L}^2} = \sqrt{\int_{\mathbb{T}^d} |f(x)|^2 dx}$ and $\|f\|_{\mathsf{L}^\infty} = \max_{x\in \mathbb{T}^d} |f(x)|$; let $\sqrt{f \,}$ denote the pointwise square root of a function $f$, and likewise, let $fg$ denote its pointwise product with another function $g$; and finally, let $\nabla^{2(s+1)}f$ denote the Laplacian $\nabla^2$ applied $s+1$ times to $f$. Next, define $\partial_x^\kappa f = \frac{\partial^{\kappa_1}}{\partial x_1^{\kappa_1}} \cdots \frac{\partial^{\kappa_d}}{\partial x_d^{\kappa_d}} f$, where the multi-index $\kappa = (\kappa_1, \dots, \kappa_d)$ is a tuple of non-negative integers. Let $\mathsf{C}^k$ denote the space of functions $f:\mathbb{T}^d \rightarrow \mathbb{C}$ for which $\partial_x^\kappa f$ is continuous for all $\|\kappa \|_1 \le k$ (including at the periodic boundaries).

\begin{theorem}[Bound on spatial discretization error]
Let $\sqrt{p_t}, \; V_t \in \mathsf{C}^{2(s+1)}$ for all $t\in [0,T]$, for some integer $s \ge (d+7)/4$. Moreover, for any $x \in \mathbb{T}^d$, let $V_t(x), \; \nabla^2 \sqrt{p_t(x)}$ and $\nabla^2 [V_t(x) \sqrt{p_t(x)}]$ be continuous in $t$ for all $t \in [0,T]$. Define
\begin{equation}
\delta = 
\big \| \ket{\psi_0} - \ket{\phi_0}\big \|
+ 
T c_s \left( \frac{Ld}{N} \right)^{2s},
\label{eq:delta}
\end{equation}
where $\| \cdot \|$ denotes the Euclidean norm, and
\begin{equation}
c_s = 3 \max_{t\in [0,T]}\bigg[ \Big( \big \| V_t \big \|_{\mathsf{L}^\infty} +1 \Big) \, \big \|\nabla^{2(s+1)}  \sqrt{p_t} \big \|_{\mathsf{L}^2}
+
\big \|\nabla^{2(s+1)} (V_t \, \sqrt{p_t}) \big \|_{\mathsf{L}^2} \bigg] + 1.
\end{equation}
If $\delta \le T$, then $\big \| \ket{\psi_T} - \ket{\phi_T}\big \| \le \delta$. 
\label{thm:spatial_disc_error}
\end{theorem}
\begin{proof}
See Appendix~\ref{app:discretizing_space}.
\end{proof}
In other words, $\delta$ from Eq.~\eqref{eq:delta} is an upper bound on the spatial discretization error, provided it is small enough. (The large error regime is harder to control and would need to be treated separately.) We can make $\delta$ small by having a small state preparation error and using a sufficiently large $N$. Notice that Theorem~\ref{thm:spatial_disc_error} gives not just one bound, but a family of bounds---one for each value of $s$---from which we are free to choose the best one. A larger $s$ gives better scaling with $N$, but requires smoother functions, and could cause $c_s$ to be larger. Also, the fact that we are quantizing flow models makes this an unusual Hamiltonian simulation problem, in that we care about the final PDF $p_T$, but we are free to choose a convenient probability path $(p_t)_{t\in[0,T]}$ that produces it. In particular, one could choose a path designed to satisfy the conditions of Theorem~\ref{thm:spatial_disc_error}, or minimize $c_s$. The probability paths in Sections~\ref{sec:flow_matching} and \ref{sec:diffusion}, for instance, are infinitely differentiable in $x$ and $t$ by construction.

In order to digitally simulate evolution by $H_t^c$, we will also need to discretize time, which is the focus of Sec.~\ref{sec:discretizing_time}. For now, we note simply that to prepare $\ket{\psi_T}$ within an error tolerance $\epsilon$, it suffices to demand that the errors from discretizing space and time both be at most $\epsilon/2$. For simplicity, suppose we can prepare the initial state perfectly ($\ket{\psi_0} = \ket{\phi_0}$), and that $T/\epsilon \ge 1/2$. Then under the conditions of Theorem~\ref{thm:spatial_disc_error}, we can ensure $\| \ket{\psi_T} - \ket{\phi_T} \| \le \epsilon/2$ by picking $N$ to be an integer power of 2 such that
\begin{equation}
N \ge L d \left( \frac{2 T c_s}{\epsilon} \right)^{1/2s},
\label{eq:N_bound}
\end{equation}
where $s \ge (d+7)/4$ is an integer of our choice. Recall that $N$ is the number of grid points along each dimension, so the number of qubits needed to represent $\ket{\psi_T}$ to the desired accuracy is
\begin{align}
n = d \log_2 (N) 
\ge 
d \log_2 \left[ Ld \left( \frac{2 T c_s}{\epsilon} \right)^{1/2s} \right].
\label{eq:dlogN_bound}
\end{align}
Observe that the number of qubits $n$ here is not a fixed quantity describing the size of the problem, but rather, an adjustable parameter that controls the spatial discretization error. Inevitably, simulating longer times $T$ requires a finer grid and therefore more qubits, since no finite-dimensional Hamiltonian $H_t^c$ can fully capture the original one $\hat{\mathcal{H}}_t^c$, and the resulting discrepancy between $\ket{\psi_t}$ and $\ket{\phi_t}$ grows with time. Note also that Eqs.~\eqref{eq:N_bound} and \eqref{eq:dlogN_bound} only describe the size of the main register in which we prepare $\ket{\psi_T}$. In order to simulate evolution by $H_t^c$, we will also need an ancilla register in which to repeatedly compute and uncompute $V_t$. We do not explicitly count these ancillas here, although we assume that enough are used that rounding errors are negligible.

\subsection{Discretizing time}
\label{sec:discretizing_time}

Define the unitary
\begin{equation}
U(t_1, t_0) = \mathcal{T} \exp \left( -i \int_{t_0}^{t_1} H_t^c \, dt \right)
\label{eq:U_time_ordered}
\end{equation}
describing evolution by the discrete continuity Hamiltonian $H_t^c$ from time $t_0$ to $t_1$, for some $0 \le t_0 \le t_1 \le T$. While we generally cannot prepare $\ket{\phi_T} = U(T,0)\ket{\phi_0}$ exactly, we can efficiently prepare a state that is within some error tolerance $\epsilon/2$ of it in Euclidean distance, and therefore within $\epsilon$ of $\ket{\psi_T}$ from Eq.~\eqref{eq:Psi_t}. We will do this using a product formula (PF) \cite{lloyd:1996}, which is the simplest possible approach. While PFs give time complexity with slightly sub-optimal scaling \cite{childs:trotter_error_comm}, they sometimes outperform more sophisticated approaches that incur larger overheads \cite{childs:PNAS}. More importantly here, PFs are conceptually straightforward to implement given a potential function learned by a classical flow model, and require no complex oracles. We therefore leave ``post-Trotter'' implementations as a subject for future work.

While several choices are possible \cite{bagherimehrab2025,chen:2022PF}, we approximate $U(t_1, t_0)$ by the simple 8-step product formula from Ref.~\cite{childs_wiebe_comm}, which concatenates two group commutators with opposite signs to achieve an $O(\Delta t^2)$ error, where $\Delta t=t_1-t_0$. Specifically, we approximate $U(t_1, t_0) \approx W(t_0)$, where
\begin{equation}
W(t) = 
e^{i \beta D_{V_t}} \;
e^{i \alpha K} \;
e^{-i \beta D_{V_t}} \;
e^{-i \alpha K} \;
e^{-i \beta D_{V_t}} \;
e^{-i \alpha K} \;
e^{i \beta D_{V_t}} \;
e^{i \alpha K}
\label{eq:PF}
\end{equation}
for Hermitian matrices $D_{V_t}$ and $K$ defined in Eq.~\eqref{eq:D_V_t} and Eqs.~\eqref{eq:F}--\eqref{eq:K_def} respectively, with angles
\begin{equation}
\alpha = \frac{L}{\pi N} \sqrt{\frac{\Delta t}{d}}
\qquad \qquad 
\beta = \frac{\pi N}{2L} \sqrt{d \, \Delta t},
\label{eq:PF_angles}
\end{equation}
where the argument in $W(t)$ specifies the time parameter in $D_{V_t}$. The approximation error is bounded in Theorem~\ref{thm:time_disc_error}.

\begin{theorem}[Bound on local time discretization error] 
Let $V_t(x)$ be continuously differentiable in $t$ for all $x\in \mathbb{T}^d$ and $t\in[0,T]$, then
\begin{equation} 
\Big \| U(t_1, t_0) - W(t_0) \Big \|
\le 
\left[ 
\frac{3 \pi^4}{4} \, \frac{d^2 N^4}{L^4} \Big(1+ \|V_{t_0}\|_{\mathsf{L}^\infty}  \Big)^4 
+ \frac{\pi^2}{2} \frac{d N^2}{L^2} \max_{t\in[t_0, t_1]} \Big\| \frac{\partial}{\partial t} V_t \Big \|_{\mathsf{L}^\infty}
\right] \Delta t^2,
\end{equation}
where $\Delta t = t_1-t_0$ and $\|\cdot \|$ denotes the spectral norm.
\label{thm:time_disc_error}
\end{theorem}
\begin{proof}
See Appendix~\ref{app:discretizing_time}.
\end{proof}

Notice that the upper bound in Theorem~\ref{thm:time_disc_error} grows with $N$, since a larger $N$ means including higher spatial frequencies in the simulation, so $\|K\| = O( N^2 d/L^2)$, and in turn $\|H_t^c\|$, grows accordingly. This is expected, since $K$ comes from discretizing $-\frac{1}{2} \nabla^2$, an unbounded operator. We therefore want the smallest possible $N$ here, while keeping the spatial discretization error below $\epsilon/2$. That entails choosing $N$ to be the smallest integer power of $2$ satisfying Eq.~\eqref{eq:N_bound}, meaning:
\begin{equation}
L d \left( \frac{2 T c_s}{\epsilon} \right)^{1/2s} \le \; \;  N \; \; \le \; \; 2 L d \left( \frac{2 T c_s}{\epsilon} \right)^{1/2s}.
\end{equation}
Finally, we can divide the interval $[0,T]$ into $r$ timesteps of duration $\Delta t= T/r$ and, by the triangle inequality, approximate $U(T,0) = U(T, T-\Delta t) \cdots U(\Delta t, 0)$ with $W_\text{tot} = W[(r-1)\Delta t]\cdots W(\Delta t)\, W(0)$, incurring a global error bounded by the sum local errors given by Theorem~\ref{thm:time_disc_error} \cite{PhysRevLett.128.210501}. To ensure $\|U(T,0)-W_\text{tot}\|\le \epsilon/2$, and in turn that $ \| \ket{\psi_T} - W_\text{tot}\ket{\psi_0}  \| \le \epsilon$, it therefore suffices to use
\begin{equation}
r \ge 
4 \pi^2 \left[ 3 \pi^2 (1+V_\text{max})^4\; d^{\,6} \left( \frac{2 T c_s}{\epsilon} \right)^{2/s} 
+
\dot{V}_\text{max} \; d^{\,3} \left( \frac{2 T c_s}{\epsilon} \right)^{1/s}
\right] \frac{T^2}{\epsilon}
\label{eq:r_bound}
\end{equation}
timesteps, where $V_\text{max} = \max_{t\in[0,T]} \|V_t\|_{\mathsf{L}^\infty}$ and $\dot{V}_\text{max} = \max_{t\in[0,T]} \|\frac{\partial}{\partial t} V_t\|_{\mathsf{L}^\infty}$.

Note that $W_\text{tot}$ involves only exponentials of $K$ and $D_{V_t}$ which are simple to implement. The latter can be implemented by computing and uncomputing $V_t$ on an ancilla register, with single-qubit rotations in between \cite{zalka1998, kassal:2008, Childs2022}. Similarly, for any real $\phi$,
\begin{equation}
e^{i \phi K} = \Big( SF \;  e^{i \phi  D_K /2} \; F^\dag S^\dag \Big) ^{\otimes d},
\end{equation}
since $K$ is defined through a Kronecker sum over $d$ registers, and $S$ and $F$ are unitary. Therefore, implementing $e^{i \phi K}$ involves single-qubit rotations (for $S$ and $S^\dag$), a quantum Fourier transform and its inverse ($F$ and $F^\dag$), and realizing $e^{i \phi D_K/2}$ as described above on each $\log_2(N)$-qubit register encoding a spatial dimension.

\section{Applications of qsample preparation for flow models}\label{sec:applications}

In practice, flow models are often trained to generate the available data, yet simply reproducing the training data does not allow for answering questions of scientific relevance. Instead, after learning a flow model, one may want to estimate the mean of a variety of functions with respect to the learned distribution (mean estimation), or generate a sample from the learned distribution with some specific property (property optimization). Said another way, after the first step of learning the flow model, one often wants to perform some sort of post-training ``inference'' task. Specific examples of different scientific contexts, and their associated inference tasks include: 
\begin{enumerate}
    \item Sampling molecular conformations via Boltzmann Generators \citep{noe2019boltzmann}. They approach sampling by learning a potentially inaccurate flow model that allows both for sampling and density evaluation. Then the learned model is used as the proposal distribution for Self-Normalized Importance Sampling to obtain a consistent sampling algorithm from the target Boltzmann density;
    \item Compositional generation \citep{du2023reduce, skreta2025feynman}, where the goal is to produce generations fitting two or more criteria at the same time, e.g., a molecule structure that can bind to several target proteins. This is usually formalized as sampling from the product of two densities given by the same flow model with different conditions;
    \item Constrained reward optimization \citep{domingo2024adjoint,singhal2025general}, i.e., generation of realistic samples that maximize specified reward function. This can be formalized as a product of the density defined by the flow model trained on realistic examples with the density proportional to the reward function;
    \item Free-energy estimation \citep{mate2024neural}---one of the fundamental problems in computational chemistry. The goal here is to estimate the ratio of normalization constants of the density defined by the flow model and restricted to two given meta-stable states. This problem is of particular importance because it allows for simulating the binding process of molecular systems.
    \item Solving inverse problems \citep{mardani2023variational}, where one is given a (usually lossy) observation process and a dataset of the original signals. One can formalize this problem as Bayesian inference, where the prior is specified by the flow model, the likelihood is the observation process, and one has to recover the posterior distribution of clean signals based on the corrupted observation.
\end{enumerate}

In the standard classical setting, one typically relies on the Monte Carlo estimates of the values of interest, or designs task-specific inference algorithms such as Sequential Monte Carlo \citep{singhal2025general,skreta2025feynman}. Naturally, this creates the bottleneck of generating individual samples from the learned model and corresponding convergence rates of Monte Carlo estimators. However, the quantum algorithm presented in the previous section for generating qsamples for a flow model creates an additional possibility; namely, using quantum algorithms based on qsamples. As such, the natural question is whether one can gain any advantages using quantum algorithms with access to qsamples, over classical algorithms with access to samples.

For the case of \textit{discrete} distributions, there is a rich history of work aimed at understanding the extent to which one may or may not gain advantages for statistical problems of the type discussed above, when one replaces access to classical samples from $p$ with some form of \textit{quantum access} to $p$~~\cite{Brassard_2002,brassard2011optimalquantumalgorithmapproximate,heinrich2001quantumsummationapplicationintegration,montanaro2015quantum,qchebyshev,qsubgaussian, Cornelissen_2022,Bravyi_2011,gilyen2020distributional,chakraborty2010newresultsquantumproperty,li2018quantum,montanaro2013survey}. Through this line of work a variety of models of ``quantum access" to a distribution have emerged, however here we highlight two such models:

\begin{definition}[qsample access]\label{def:qsample-access} A distribution $p:\Omega\rightarrow[0,1]$ over some finite set $\Omega$, is accessible via qsample access if one has access to copies of the qsample state
\begin{equation}\label{eq:qsamp}
|\psi_p\rangle = \sum_{x\in\Omega}\sqrt{p(x)}|x\rangle.
\end{equation}
\end{definition}

\begin{definition}[qsample preparation access]\label{def:qsample-prep-access} A distribution $p:\Omega\rightarrow[0,1]$ over some finite set $\Omega$, is accessible via qsample preparation access if one has the ability to implement both $U_p$ and $U^{-1}_p$ for a unitary $U_p$ satisfying
\begin{equation}\label{eq:qsamp-prep}
U_p|0\rangle = |\psi_p\rangle
\end{equation}
for some fiducial state $\ket{0}$. In this model, one implementation of either $U_p$ or $U^{{-1}}_p$ is considered to be a ``query'' to the unknown distribution $p$.
\end{definition}

While qsample access may seem like the natural quantization of classical samples, most works on quantum algorithms for ``distribution problems'' in fact make the stronger assumption of qsample \textit{preparation} access. Indeed, with this assumption one can show that, at least for discrete distributions, quantum algorithms can obtain meaningful query complexity advantages over classical algorithms for both mean estimation, and a variety of property testing tasks.

More specifically, for mean estimation with such quantum access there is a long line of works on advantageous quantum mean estimators under different assumptions~\cite{Brassard_2002,brassard2011optimalquantumalgorithmapproximate,heinrich2001quantumsummationapplicationintegration,montanaro2015quantum,qchebyshev,qsubgaussian,kothari2023mean,Cornelissen_2022}, culminating in the quantum sub-Gaussian estimators of Ref.~\cite{qsubgaussian} and Ref.~\cite{Cornelissen_2022} (the latter for the estimation of the mean of multi-variate random variables). For property testing, advantageous quantum algorithms have been developed for testing uniformity and closeness between distributions~\cite{Bravyi_2011,gilyen2020distributional}, certifying distributions (i.e., testing identity)~\cite{chakraborty2010newresultsquantumproperty}, and testing entropy~\cite{li2018quantum,gilyen2020distributional} from quantum access (see also~\cite{montanaro2013survey} for a review). We stress, however, that all of the above algorithms (and the advantages that they are able to obtain) are for the case of \textit{discrete} distributions, and to the best of our knowledge, the extent to which one can or cannot gain advantages for distribution problems in the continuous setting is as of yet unexplored. 

But as we have already noted above, the quantum algorithm for simulating the continuity equation given in Section~\ref{sec:qsim-algorithm} allows us to prepare \textit{approximate discretized} qsamples for flow models---i.e., for a continuous distribution which is the solution to the continuity equation for some known velocity field $v_t = \nabla V_t$---and therefore provides a concrete and efficient way to realize a specific type of quantum access to such distributions. Motivated by this, and by the myriad of statistical ``inference'' problems defined by flow models of scientific relevance, in the following section we formally define the notion of ``approximate discretized qsample preparation access'' which  can be realized through the quantum algorithm for simulating the continuity equation given in Section~\ref{sec:qsim-algorithm}. With this in hand, we then show how approximate discretized qsample preparation access allows one to gain rigorous query complexity advantages over classical algorithms for mean estimation. This then directly implies advantages for mean estimation of functions with respect to flow models, over standard naive approaches based simply on collecting and post-processing classical samples from the flow model.

\subsection{Approximate discretized qsample preparation access to continuous distributions}\label{ss:disc_approx_qsample_preparation}

As done previously in this work, we consider probability density functions $p:[0,L]^d\rightarrow [0,\infty)$. In order to define the notion of approximate discretized  qsample preparation access to some such PDF $p$, we start by defining a canonical discretization. To do this, we follow the approach used in Section~\ref{sec:discretizing_space}, and start by defining the 1-dimensional grid
\begin{equation}
\mathbb{X}_N = \frac{L}{N} \, \Big \{ 0, 1, \dots, N-1 \Big\} \; \subset [0,L]^d.
\label{eq:X_def}
\end{equation}
Then $\mathbb{X}_N^d \subset [0,L]^d$  is the $d$-dimensional grid comprising $N^d$ points covering $[0,L]^d$ uniformly, with spacing $L/N$ along each spatial dimension. Given a  PDF $p:[0,L]^d\rightarrow [0,\infty)$ such that $p(x) < \infty$ for all $ x \in \mathbb{X}_N^d$ and $p(x)>0$ for some $ x \in \mathbb{X}_N^d$, we then define its discretization $\bar{p}_N:\mathbb{X}^d_N\rightarrow[0,1]$ via: 
\begin{equation}
\bar{p}_N(x) = \frac{p(x)}{S} \text{ for all } x\in\mathbb{X}_N^d,
\end{equation}
where $S$ is a normalization constant necessary to ensure that $\bar{p}_N$ is a valid probability mass function (PMF). 

With this in hand, recall from Eqs.~\eqref{eq:qsamp} and~\eqref{eq:qsamp-prep} that $U_{\bar{p}_N}$ is the unitary defined via
\begin{equation}
U_{\bar{p}_N}|0\rangle = |\psi_{\bar{p}_N}\rangle = \sum_{x\in\mathbb{X}_N^d}\sqrt{\bar{p}_N(x)}|x\rangle.
\end{equation}
We can then define the notion of discretized qsample preparation access to $p$ in a natural way. 

\begin{definition}[Discretized qsample preparation access]\label{def:disc-qsample-prep-access} Let $p:[0,L]^d\rightarrow[0,\infty)$ be a probability density function. We say that the distribution $p$ is accessible via discretized qsample preparation access with discretization scale $N$ if one has the ability to implement both $U_{\bar{p}_N}$ and $U^{-1}_{\bar{p}_N}$.
\end{definition}

With the above definition established, it should now be clear that the quantum algorithm presented in Section~\ref{sec:qsim-algorithm} \textit{almost} allows one to efficiently realize discretized qsample preparation access for any PDF which is the solution to the continuity equation for some known conservative velocity field $v_t$. Unfortunately however, the simulation algorithm from Section~\ref{sec:qsim-algorithm} only \textit{approximately} implements $U_{\bar{p}_N}$ (i.e., up to an arbitrarily small error). As such, we finally define the notion of \textit{approximate} discretized qsample preparation access:

\begin{definition}[Approximate discretized qsample preparation access]\label{def:app-disc-qsample-prep-access} Let $p:[0,L]^d\rightarrow[0,\infty)$ be a probability density function. We say that the distribution $p$ is accessible via $\epsilon$-approximate discretized qsample preparation access with discretization scale $N$ if one has the ability to implement both $U$ and $U^{-1}$ for some unitary $U$ satisfying ${\|U|0\rangle - |\psi_{\bar{p}_N}\rangle\| \leq \epsilon}$, where $\| \cdot \|$ denotes the Euclidean/2-norm.
\end{definition}

This notion is now precisely the type of access that can be concretely and efficiently realized by the algorithm given in Section~\ref{sec:qsim-algorithm}, for any conservative flow model---i.e., for any PDF which is the solution to the continuity equation for some known conservative velocity field $v_t$ obeying appropriate regularity conditions. Also, we note that the $2$-norm appearing in Definition~\ref{def:app-disc-qsample-prep-access} is inherited from the guarantees given on the algorithm described in Section~\ref{sec:qsim-algorithm}. While this might appear strange, this norm upper bounds the trace distance between the states $U|0\rangle$ and $|\psi_{\bar{p}_N}\rangle$ \cite{aharonov:mixed}, which in turn upper bounds the total variation distance between the output distributions obtained from measuring these states in the computational basis---i.e., between $\bar{p}_N$ and the Born distribution of $U|0\rangle$, which are perhaps the more natural objects to consider. Given this, the immediate question is whether or not quantum algorithms with such quantum access to a PDF $p$ can gain any advantages for natural ``inference tasks''---i.e., statistical problems defined by flow models---over classical algorithms with standard classical sample access. In the following section we show that this is indeed the case for the ubiquitous task of \textit{mean estimation}. 

\subsection{Advantages in continuous mean estimation with approximate discretized qsample preparation access}\label{ss:app-mean-estimation}

To illustrate the power of the quantum distribution access models defined above, here we consider the problem of \textit{mean estimation} of a random variable with respect to an underlying probability distribution. More specifically, given some function ${f:\Omega\rightarrow E\subset \mathbb{R}}$, together with a probability mass function ${p: \Omega\rightarrow[0,1]}$ (when $\Omega$ is a finite set) or probability density function $p:\Omega\rightarrow[0,\infty)$ (when $\Omega$ is an infinite set), we are interested in the problem of estimating the mean $\mu = \mathbb{E}_{x\sim p}[f(x)]$. In the classical case we assume that we are given independent and identically distributed (i.i.d.) samples from $p$, and in the quantum case we assume that we have some sort of quantum access to $p$ (such as those defined in Definitions~\ref{def:qsample-access}-\ref{def:app-disc-qsample-prep-access}) as well as some sort of quantum access to $f$. More specifically, in the case when $\Omega$ is finite, one typically assumes that one has access to an ``evaluation oracle'' for $f$, i.e., the unitary satisfying $U_f |x,0\rangle = |x,f(x)\rangle$. We note that one can construct an explicit and efficient implementation of such an oracle whenever an efficient classical circuit for computing $f$ is known~\cite{nielsen2010quantum}.

We are then interested in obtaining \textit{high probability deviation bounds} for different estimators, which are bounds of the following type: Given $m$ ``experiments'' and an allowed failure probability $\Delta\in (0,1)$, what is the smallest error $\epsilon(m,\Delta,f)$ such that the estimator $\tilde{\mu}$ satisfies $|\tilde{\mu}-\mu| >\epsilon(m,\Delta,f)$ with probability at most $\Delta$? In the classical case, an ``experiment'' is simply a single sample from $p$, and in the quantum case, an ``experiment'' is any single implementation of the evaluation oracle $U_f$, or any of the unitaries appearing in the quantum access models in Definitions~\ref{def:qsample-access}-\ref{def:app-disc-qsample-prep-access}.

Classically, the best estimators are the \textit{sub-Gaussian} estimators (such as ``median-of-means''), which obtain 
\begin{equation}\label{eq:classical-sub-gaussian}
\mathrm{Pr}\left[|\tilde{\mu}-\mu|> C\sigma\sqrt{\frac{\log(1/\Delta)}{m}}\right]\leq \Delta
\end{equation}
for some constant $C$, where $\sigma = \sqrt{\text{Var}_{x \sim p} f(x)}$. From a quantum perspective, for the case of discrete underlying distributions $p$, there is a long line of work on quantum mean estimators, under different assumptions on both $p$ and $f~$\cite{Brassard_2002,brassard2011optimalquantumalgorithmapproximate,heinrich2001quantumsummationapplicationintegration,montanaro2015quantum,qchebyshev,qsubgaussian}. These works have culminated in the quantum sub-Gaussian estimator of Ref.~\cite{qsubgaussian}. In particular, prior to the results of Ref.~\cite{qsubgaussian} all quantum mean estimators either required a strong promise on the variance, or are less efficient than classical sub-Gaussian estimators for heavy-tailed distributions. With this in mind, the state-of-the art result from Ref.~\cite{qsubgaussian} is as follows:

\begin{theorem}[Quantum sub-Gaussian mean estimator -- Theorem 4.2 Ref.~\cite{qsubgaussian}]\label{thm:discreteq} Assume some underlying PMF $p:\Omega\rightarrow[0,1]$ and function $f:\Omega\rightarrow E\subset \mathbb{R}$. There exists a quantum sub-Gaussian estimator which, when given qsample preparation access to the unknown PMF $p$ (as per Definition~\ref{def:qsample-prep-access}), together with access to the evaluation oracle $U_f$, as well as $m$ and $\Delta\in(0,1)$ such that $m\geq \log(1/\Delta)$, uses $O(m\log^{3/2}(m)\log\log (m))$ ``experiments'' and outputs a mean estimate $\tilde{\mu}$ satisfying
\begin{equation}
\mathrm{Pr}\left[|\tilde{\mu}-\mu|> \sigma\frac{\log(1/\Delta)}{m}\right]\leq \Delta.
\end{equation}
\end{theorem}
As discussed in Ref.~\cite{qsubgaussian}, this essentially gives a quadratic speed-up over the number of classical i.i.d.\ samples necessary to estimate the mean of heavy-tailed distributions with a sub-Gaussian error rate, without requiring any information on the variance beyond a promise that it is finite. We note that the estimator described above has since been extended to a quantum sub-Gaussian estimator for \textit{multivariate} mean estimation~\cite{Cornelissen_2022}---i.e., for ${f:\Omega\rightarrow E\subset \mathbb{R}^d}$---however here we restrict ourselves to univariate mean estimation for ease of presentation.

As mentioned earlier, here we are now particularly interested in mean estimation in the case when the underlying distribution is \textit{continuous}, i.e., when $\Omega = [0,L]^d$ and the distribution is described by some PDF ${p:[0,L]^d\rightarrow [0,\infty)}$. Additionally, we are particularly interested in the setting where one is given approximate discretized qsample preparation access to $p$ (as per Definition~\ref{def:app-disc-qsample-prep-access}). Again, this is motivated by the fact that such quantum access can be concretely and efficiently realized by the quantum simulation algorithm described in Section~\ref{sec:qsim-algorithm}. 

Here there is a particularly simple and natural strategy: simply apply the \textit{discrete} quantum mean estimator of Theorem~\ref{thm:discreteq}, ignoring the fact that the unknown underlying distribution is continuous. This works, as the access model defined in Definition~\ref{def:app-disc-qsample-prep-access} gives one the ability to apply the unitary for the preparation of a qsample for a \textit{discrete} distribution, which is an approximation of the true underlying continuous distribution. As such, one effectively performs mean estimation with respect to the approximate distribution, and all that is necessary to analyze is the effect of the approximations in the access model on the performance of the resulting estimator. If one does this, one obtains the following result, which is essentially a corollary of Theorem~\ref{thm:discreteq}.

\begin{corollary}[Quantum mean estimation for continuous distributions]\label{c:qmeansapprox} Assume a PDF $p:[0,L]^d\rightarrow[0,\infty)$ which is Lipschitz with Lipschitz constant $\ell_p$, and some function $f:[0,L]^d\rightarrow E\subset \mathbb{R}$ which is Lipschitz with Lipschitz constant $\ell_f$.
The quantum sub-Gaussian estimator of Theorem~\ref{thm:discreteq}, when given $\epsilon$-approximate discretized qsample preparation access to $p$ with discretization scale $N$ (as per Definition~\ref{def:app-disc-qsample-prep-access}), together with evaluation oracle access $U_f$ to the restriction of $f$ on $\mathbb{X}_N^d$, as well as $m$ and $\Delta\in(0,1)$ such that $m\geq \log(1/\Delta)$, uses $O(m\log^{3/2}(m)\log\log (m))$ ``experiments'' and outputs a mean estimate $\tilde{\mu}$ satisfying
\begin{equation}
\mathrm{Pr}\left(|\tilde{\mu}-\mu|> \sqrt{\sigma^2 + \epsilon_{\mathrm{var}}}\; \frac{\log(1/\Delta)}{m} + \epsilon_{\mathrm{mean}}\right) \leq \Delta.
\end{equation}
where $\sigma = \sqrt{\mathrm{Var}_{x \sim p} f(x)}$ and
\begin{align}
\epsilon_{\mathrm{mean}} &= 2\|f\|_{\mathsf{L}^\infty} \frac{\ell_p\sqrt{d}L^{d+1}}{N} + \ell_f \sqrt{d}\frac{L}{N}  + 2\epsilon\|f\|_{\mathsf{L}^\infty} , \label{eq:mean-error-cor} \\
\epsilon_{\mathrm{var}} &=  6\|f\|^2_{\mathsf{L}^\infty} \frac{\ell_p\sqrt{d}L^{d+1}}{N} + 4 \| f\|_{\mathsf{L}^\infty} \; \ell_f \sqrt{d}\frac{L}{N} + 6\epsilon \|f\|^2_{\mathsf{L}^\infty}. \label{eq:variance-error-cor} 
\end{align}
\end{corollary}
\begin{proof}
See Appendix~\ref{app:proof-cor-mean-estimation}.
\end{proof}
By comparing with the performance of the classical sub-Gaussian estimator given in Eq.~\eqref{eq:classical-sub-gaussian}, the take-away from this result is that when one takes into account the errors from both discretization and approximation, one can still achieve advantages quantumly for \textit{continuous} mean estimation with approximate discretized qsample preparation access, but only in certain parameter ranges, and in particular not for arbitrarily small target error. However, as is clear from Eqs.~\eqref{eq:mean-error-cor} and~\eqref{eq:variance-error-cor}, one can systematically reduce these errors (and the ranges in which one can obtain quantum advantages for mean estimation) by increasing the discretization scale $N$ or decreasing the approximation error $\epsilon$ of the assumed access model. The quantum resources (qubits, time etc.) needed to do this have been detailed precisely in Section~\ref{sec:qsim-algorithm}.

\section{Conclusion}

In recent years, flow models have emerged as a state-of-the-art method for learning and sampling from complex, high-dimensional probability distributions. Indeed, they have had remarkable success in learning complex distributions of practical relevance across a wide variety of domains, and remain an extremely active field of research~\cite{lipman2024flowmatchingguidecode}. In this work, we have constructed a bridge between flow models and quantum computing by showing in Section~\ref{sec:flow_quantum_mapping} that, given a trained flow model, one can construct a specific (continuous variable) Hamiltonian whose dynamics prepare a qsample, i.e., a coherent encoding, of the distribution learned by the flow model. Accordingly, we refer to time evolution under this Hamiltonian as a \textit{wavefunction flow}.

Motivated by this connection, we have developed in Section~\ref{sec:qsim-algorithm} a method for efficient quantum simulation of this specific class of Hamiltonians on digital quantum computers; i.e., a method to efficiently implement wavefunction flows. This immediately yields an efficient quantum algorithm for preparing qsamples for any distribution described by a flow model (satisfying mild regularity conditions)---a class of distributions which has been shown empirically to be extremely expressive. As such, this algorithm can vastly enlarge the class of distributions for which the efficient preparation of qsamples is possible. 

This fundamental connection between flow models and quantum computing enables a variety of new perspectives and opportunities. On one hand, from an ML perspective, it offers a different---and potentially more powerful---type of access to the distributions described by flow models. Indeed, training a flow model is often a first step towards probing properties of the learned distribution, and in Section~\ref{sec:applications}, building on prior work~\cite{qsubgaussian}, we have demonstrated that the expected value of any function with respect to the learned distribution can be estimated from fewer qsamples than standard classical samples. For the case of discrete distributions, it is known that qsamples also offer advantages for a wide variety of property testing tasks, such as identity testing~\cite{gilyen2020distributional}. We therefore leave as an interesting open direction for future research to uncover further advantages of qsamples for the continuous distributions described by flow models.

From a quantum computing perspective, it is well known that certain complexity-theoretic conjectures can be reduced to questions concerning the complexity of preparing qsamples for certain classes of distributions~\cite{aharonov2003adiabaticquantumstategeneration}. As this work has shown that qsamples for distributions learned by flow models can be efficiently prepared, these complexity theoretic conjectures could be recast as questions concerning the limitations of flow models. This work therefore brings these quantum complexity-theoretic questions firmly into the domain of mainstream research on the mathematical foundations of modern ML methods. Specific potential questions include:
\begin{question}[Hardness of Sampling]
Is there a velocity field $v_t$ such that $p(x) = |\langle x | \psi_T \rangle|^2$ is a classically hard-to-sample distribution under standard computational complexity conjectures? If so, can such a $v_t$ be learned efficiently?
\end{question}

\begin{question}[Quantum Circuits and Wavefunction Flows]
    Let $U$ be a quantum circuit that produces a qsample $\ket{\psi_T}$ when applied to a fiducial state. Is there an algorithm that, given the circuit description, produces $v_t$ and a corresponding continuity Hamiltonian \eqref{eq:continuity_Hamiltonian} whose dynamics \eqref{eq:U_time_ordered} approximate $U$? Is it possible to show, or refute, that efficient simulation of the continuity Hamiltonian is a $\textsf{BQP}$-complete task?     
\end{question}

In summary, we have attempted in this work to construct a meaningful bridge between state-of-the-art ML methods and quantum computing, and to explore some of its immediate consequences. Historically, novel connections between physics and ML have proven extremely fruitful for both sides, and we hope that this work can continue this tradition by catalyzing the development of new methods and insights in both ML and physics.

\section*{Acknowledgments}

We thank Srinivasan Arunachalam, Arkopal Dutt, Hari Krovi, Bryan O'Gorman, and Oles Shtanko for helpful discussions.

\bibliography{references}

\appendix
\section*{Appendix}
\setcounter{theorem}{0}
\renewcommand{\thetheorem}{\Alph{section}\arabic{theorem}}

\section{Continuity Hamiltonian}
\label{app:continuity_H}

In deriving the continuity Hamiltonian $\hat{\mathcal{H}}_t$, Eq.~\eqref{eq:schrodinger_derivation} divides by $\sqrt{p_t(x)}$ and therefore tacitly assumes that $p_t$ is positive everywhere. Here we show that this assumption is not necessary. Suppose $(p_t)_{t\in[0,T]}$ obeys the continuity equation \eqref{eq:continuity_eq} with respect to a velocity field $v_t$. Moreover, suppose a generic wavefunction $\Psi_t$ satisfies the Schr\"odinger equation $i \frac{\partial}{\partial t} \Psi_t = \hat{\mathcal{H}}_t \Psi_t$ with initial condition $\Psi_0(x) = \sqrt{p_0(x)}$, then the function $(x,t) \mapsto \Psi_t(x)^2$ (note the lack of absolute value) also obeys the continuity equation with respect to $v_t$:
\begin{align}
\frac{\partial}{\partial t} \Psi_t(x)^2
& =
-2i \Psi_t(x) \; \hat{\mathcal{H}}_t \Psi_t (x) \nonumber \\ 
& =
- \Psi_t(x) \; \nabla \cdot \big[ v_t(x) \Psi_t(x) \big]
- 
\big[ v_t(x) \Psi_t(x) \big] \cdot \nabla \Psi_t(x) \\ 
& = 
- \nabla \cdot  \big[ v_t(x) \Psi_t(x)^2 \big] \nonumber
\end{align}
with the same initial condition. Therefore, $\Psi_t(x) = \sqrt{p_t(x)}$ for $t\in[0,T]$. Notice that we need not assume $p_t(x)>0$.

\section{Hamiltonian simulation details---discretizing space}
\label{app:discretizing_space}

An absolutely integrable function $f : \mathbb{T}^d \rightarrow \mathbb{C}$ has Fourier coefficients
\begin{equation}
\widehat{f}\, (k)
=
\frac{1}{L^d} \int_{\mathbb{T}^d} f(x) \, e^{-i k \cdot x} dx
\end{equation}
for wave vectors $k \in \frac{2 \pi}{L} \mathbb{Z}^d$, and a corresponding Fourier series:
\begin{equation}
f(x) \sim \sum_{k \in \frac{2\pi}{L} \mathbb{Z}^d} \widehat{f}\, (k) \, e^{i k \cdot x}.
\end{equation}
If $f$ is also square integrable, then the squared sum of its Fourier coefficients is related to its $\mathsf{L}^2$ norm by Parseval's theorem:
\begin{equation}
\|f\|_{\mathsf{L}^2}^2
=
L^d \sum_{k \in \frac{2\pi}{L} \mathbb{Z}^d} \big| \widehat{f}(k) \big|^2.
\end{equation}

If $f \in \mathsf{C}^a$ for $a >d/2$, then its Fourier series converges both absolutely and uniformly \cite{taylor:2010} (\S3.1). Intuitively, the reason why $f$ must be smoother in higher dimensions is that its Fourier series has more terms with $\| k \| \le \text{const.}$ when $d$ is larger, so $\widehat{f}(k)$ must decay faster with $\|k\|$ for the series to converge. But the rate at which $\widehat{f}(k)$ decays with growing $\|k\|$ is closely related to the smoothness of $f$ \cite{grafakos:2014}.

In order to prove Theorem~\ref{thm:spatial_disc_error} and the supporting lemmas, it will be useful to define
\begin{equation}
\mathbb{K}_N = \frac{2\pi}{L} \Big \{ -\frac{N}{2}, -\frac{N}{2}+1, \dots, \frac{N}{2}-1\Big\},
\end{equation}
so that $\mathbb{K}_N^d$ is a set of $N^d$ wave vectors centered (roughly) around $k=(0,\dots,0)$. (We will sometimes denote $k=(0,\dots,0)$ just as $k=0$.) This lets us define the space of bandlimited functions from $\mathbb{T}^d$ to $\mathbb{C}$:
\begin{equation}
\mathsf{B} = \text{span} \big \{ \varphi_k \,  | \, k\in \mathbb{K}_N^d \big \} \;\; \subset \mathsf{C}^\infty,
\end{equation}
where $\varphi_k(x)=e^{i k \cdot x}$ is a plane wave with wave vector $k$. The main idea of the following proofs is to project the continuous-space Schr\"odinger equation onto $\mathsf{B}$, and bound the impact of doing so. Functions $f\in \mathsf{B}$ can be fully described using $N^d$ coordinates, and there is a simple (exact) expression for $\nabla^2 f$ in terms of these coordinates, so such functions lend themselves well to discretization. Of course, the original Schr\"odinger equation involves components both within $\mathsf{B}$ and outside it, so we can bound discretization error by bounding the latter components.

One might expect that the natural way to project a function $f : \mathbb{T}^d \rightarrow \mathbb{C}$ onto $\mathsf{B}$ is to truncate its Fourier series by discarding $\widehat{f}(k)$ for $k \notin \mathbb{K}_N^d$, while keeping the low-frequency coefficients unchanged. This is a common approach in other contexts, but it would give a more complicated discrete Hamiltonian that is harder to simulate \cite{lubich2008quantum, su:2021}. Instead, the proofs below use the oblique (i.e., not orthogonal) projector $\mathcal{P}$ defined by
\begin{equation}
(\mathcal{P}f)(x)
=
\frac{1}{N^d} 
\sum_{k \in \mathbb{K}_N^d} \sum_{x' \in \mathbb{X}_N^d} f(x') \,  e^{i k\cdot (x-x')}.
\label{eq:P_def}
\end{equation}
Since $\mathcal{P}f\in \mathsf{B}$ manifestly, its Fourier coefficients $\widehat{\mathcal{P}f}(k)$ vanish for $k \notin \mathbb{K}_N^d$, but they do not exactly coincide with $\widehat{f}(k)$ for $k\in \mathbb{K}_N^d$. Rather, $\mathcal{P}$ is constructed so that $(\mathcal{P}f)(x)=f(x)$ for all grid points $x\in \mathbb{X}_N^d$, even though $\mathcal{P}f \neq f$ in general. That is, we construct $\mathcal{P}f$ to agree with $f$ at grid points $x\in \mathbb{X}_N^d$ in position space, rather than at $k \in \mathbb{K}_N^d$ in reciprocal space.

We will group the supporting lemmas here into three sections: first, those that characterize functions within $\mathsf{B}$, then those characterizing the projector $\mathcal{P}$ onto $\mathsf{B}$, and finally, those bounding the solutions of differential equations for finite-dimensional vectors, which is how we will ultimately prove Theorem~\ref{thm:spatial_disc_error}.

\subsection{Characterizing functions in $\mathsf{B}$}

\begin{lemma}
Let $f \in \mathsf{B}$, then 
\begin{equation}
\sqrt{\sum_{x \in \mathbb{X}_N^d} \big |f(x) \big|^2} = \left(\frac{N}{L}\right)^{d/2} \big \|f  \big\|_{\mathsf{L}^2}.
\end{equation}
\label{lemma:parseval_X}
\end{lemma}

\begin{proof}
Since $f \in \mathsf{B}$, we can express it as a Fourier series with a finite number of terms:
\begin{equation}
f(x)=\sum_{k \in \mathbb{K}_N^d} \widehat{f}(k)\,  e^{i k\cdot x},
\end{equation}
so
\begin{equation}
\sum_{x\in \mathbb{X}_N^d} \big| f(x) \big|^2
=
\sum_{k, k' \in \mathbb{K}_N^d} \widehat{f}(k') \widehat{f}(k)^* \sum_{x \in \mathbb{X}_N^d} e^{i (k'-k) \cdot x},
\end{equation}
where the star denotes a complex conjugate. We can evaluate the last sum by expressing it as a finite geometric series in each dimension (see the proof of  Lemma~\ref{lemma:aliasing} for details) to get: 
\begin{equation}
\sum_{x\in \mathbb{X}_N^d} e^{i (k'-k)\cdot x}
=
\begin{cases}
N^d & \text{if } (k'-k) \in  \frac{2\pi N}{L} \mathbb{Z}^d \\ 
0 & \text{otherwise}
\end{cases}
\;=\;
N^d \, \delta_{k,k'},
\end{equation}
since $k,k' \in \mathbb{K}_N^d$. Therefore, by Parseval's theorem:
\begin{equation}
\sum_{x\in \mathbb{X}_N^d} \big| f(x) \big|^2
= 
N^d \sum_{k, k' \in \mathbb{K}_N^d} \widehat{f}(k') \widehat{f}(k)^* \, \delta_{k,k'}
=
N^d \sum_{k \in \mathbb{K}_N^d} \big| \widehat{f}(k) \big|^2
=
\left(\frac{N}{L} \right)^d \big \|f  \big\|_{\mathsf{L}^2}^2.
\end{equation}
\end{proof}

\begin{lemma}
Let $f\in \mathsf{B}$, then for the $N^d \times N^d$-dimensional matrix $K$ defined in Eq.~\eqref{eq:K_def}:
\begin{equation}
-\frac{1}{2} \left[ \sum_{x \in \mathbb{X}_N^d} \big (\nabla^2 f \big)(x)  \, \ket{x} \right]
=
K \left[ \sum_{x \in \mathbb{X}_N^d} f(x)  \, \ket{x} \right] .
\end{equation}
\label{lemma:laplacian}
\end{lemma}

\begin{proof}
From the definitions of $F$, $S$, $D_K$ in Eqs.~\eqref{eq:F}--\eqref{eq:D_K},
\begin{equation}
SFD_K F^\dag S^\dag = 
\frac{1}{N} \left( \frac{2\pi}{L} \right)^2 \; 
\sum_{j,\ell=0}^{N-1} \left[ \sum_{k=-N/2}^{N/2-1} k^2 e^{i 2\pi k(j-\ell)/N} \right] \ket{j} \! \bra{\ell}.
\end{equation}
Notice that the inner sum (over $k$) would simplify greatly if we removed the factor of $k^2$. That is, for $0 \le j, \ell \le N-1$:
\begin{equation}
\sum_{k=-N/2}^{N/2-1} e^{i 2\pi k(j-\ell)/N}
=
e^{i \pi (\ell-j)}  \sum_{k=0}^{N-1} \Big( e^{i 2\pi (j-\ell)/N} \Big)^k
=
N \delta_{j \ell},
\end{equation}
so we can express the $N\times N$ identity matrix $I$ as
\begin{equation}
I = \frac{1}{N} \sum_{j,\ell=0}^{N-1} \left[ \sum_{k=-N/2}^{N/2-1} e^{i 2\pi k(j-\ell)/N} \right] \ket{j} \! \bra{\ell} .
\end{equation}
Therefore, the first term of $K$ can be written as
\begin{align}
\frac{1}{2}\big( SFD_K F^\dag S^\dag \big) \otimes I^{\otimes (d-1)}
=
\frac{1}{2N^d} \left( \frac{2\pi}{L} \right)^2 \; 
\sum_{j, \ell \in \{0, \dots, N-1\}^d } \left[
\sum_{k \in \{-\frac{N}{2}, \dots, \frac{N}{2}-1\}^d}
k_1^2 \; e^{i 2\pi k \cdot (j-\ell)/N} 
\right] \ket{j}\!\bra{\ell},
\end{align}
or using the notation from Eq.~\eqref{eq:|x>_notation}:
\begin{align}
\frac{1}{2}
\big( SFD_K F^\dag S^\dag \big) \otimes I^{\otimes (d-1)}
=
\frac{1}{2N^d} \sum_{x, x' \in \mathbb{X}_N^d} \; \sum_{k \in \mathbb{K}_N^d}
\; k_1^2 \,  e^{i k \cdot (x-x')} \ket{x} \!\bra{x'}.
\end{align}
An analogous result holds for every term in the Kronecker sum $\big( SFD_K F^\dag S^\dag \big)^{\oplus d}$, just with $k_1^2 \rightarrow k_j^2$ in the $j^\text{th}$ term, so
\begin{equation}
K
=
\frac{1}{2N^d} \sum_{j=1}^d \; \sum_{x, x' \in \mathbb{X}_N^d} \; \sum_{k \in \mathbb{K}_N^d}
\; k_j^2 \,  e^{i k \cdot (x-x')} \ket{x} \!\bra{x'}
=
\frac{1}{2N^d} \sum_{x, x' \in \mathbb{X}_N^d} \; \sum_{k \in \mathbb{K}_N^d}
\; \|k\|^2 \,  e^{i k \cdot (x-x')} \ket{x} \!\bra{x'}.
\end{equation}
Finally, we can express any function $f\in \mathsf{B}$ in terms of its nontrivial Fourier coefficients $\widehat{f}(k)$ as
\begin{equation}
f(x) = \sum_{k \in \mathbb{K}_N^d} \widehat{f}(k) \, e^{i k \cdot x}
\end{equation}
for all $x \in \mathbb{T}^d$. Equivalently, we can express these coefficients in terms of $f$ evaluated on grid points $x\in \mathbb{X}_N^d$ as
\begin{equation}
\widehat{f}(k)
=
\frac{1}{N^d} \sum_{x \in \mathbb{X}_N^d} f(x) \, e^{-i k \cdot x}
\end{equation}
for $k \in \mathbb{K}_N^d$. Therefore,
\begin{align}
-\frac{1}{2} \nabla^2 f(x) 
=  
\frac{1}{2} \sum_{k \in \mathbb{K}_N^d} \|k\|^2 \widehat{f}(k) \, e^{i k \cdot x}
=
\frac{1}{2N^d}\sum_{x' \in \mathbb{X}_N^d} \sum_{k \in \mathbb{K}_N^d}  \|k\|^2 e^{i k \cdot (x-x')} f(x')
\end{align}
for any $x \in \mathbb{T}^d$. Evaluating this equation at each grid point $x \in \mathbb{X}_N^d$ and writing the resulting linear system in vector form gives the desired result:
\begin{equation}
-\frac{1}{2} \sum_{x \in \mathbb{X}_N^d} (\nabla^2 f)(x) \ket{x}
=
\frac{1}{2N^d} \sum_{k \in \mathbb{K}_N^d} \sum_{x, x' \in \mathbb{X}_N^d} \|k\|^2 e^{i k \cdot (x-x')} f(x') \ket{x} 
=
K \sum_{x \in \mathbb{X}_N^d} f(x) \ket{x}.
\end{equation}
\end{proof}

\subsection{Characterizing the projector $\mathcal{P}$ onto $\mathsf{B}$}

\begin{lemma}[Aliasing formula] Let $f \in \mathsf{C}^a$ for $a > d/2$, then the nontrivial Fourier coefficients of $\mathcal{P}f$ are related to those of $f$ by
\begin{equation}
\widehat{\mathcal{P}f} (k) = \sum_{k' \in \frac{2\pi}{L} \mathbb{Z}^d} \widehat{f} \, (k + N k'), \qquad \mathrm{for}\text{ } k\in \mathbb{K}_N^d.
\label{eq:aliasing}
\end{equation}
\label{lemma:aliasing}
\end{lemma}

\begin{proof}
Notice from the definition of $\mathcal{P}$ in Eq.~\eqref{eq:P_def} that the nontrivial Fourier coefficients of $\mathcal{P}f$ are
\begin{equation}
\widehat{\mathcal{P}f}(k)
=\frac{1}{N^d} \sum_{x' \in \mathbb{X}_N^d} f(x') e^{-i k\cdot x'}, \qquad k \in \mathbb{K}_N^d.
\end{equation}
Since $f \in \mathsf{C}^a$ for $a >d/2$, the Fourier series of $f$ converges uniformly, so for every grid point $x' \in \mathbb{X}_N^d$ we have
\begin{equation}
f(x')=\sum_{k' \in \frac{2\pi}{L} \mathbb{Z}^d} \widehat{f}\,(k') e^{i k' \cdot x'}.
\end{equation}
Substituting this equation into the preceding one gives
\begin{equation}
\widehat{\mathcal{P}f}(k) 
=
\sum_{k' \in \frac{2\pi}{L} \mathbb{Z}^d} \widehat{f}\,(k') 
\left[  \frac{1}{N^d} \sum_{x' \in \mathbb{X}_N^d} e^{i (k'-k) \cdot x'}
\right]
=
\sum_{k' \in \frac{2\pi}{L} \mathbb{Z}^d} \widehat{f}\,(k') 
\prod_{j=1}^d \left[ \frac{1}{N} \sum_{\ell=0}^{N-1} e^{i(k_j'-k_j) \ell L/N }\right].
\label{eq:aliasing_sum}
\end{equation}
Notice that $e^{i(k_j'-k_j) L/N} = 1$ if and only if $(k_j-k_j')\in \frac{2\pi N}{L} \mathbb{Z}$, so the last term in square brackets evaluates to
\begin{equation}
\frac{1}{N} \sum_{\ell=0}^{N-1} \left( e^{i(k_j'-k_j) L/N } \right)^\ell
=
\begin{cases}
1 & \text{if } (k_j-k_j')\in \frac{2\pi N}{L} \mathbb{Z} \\
\frac{e^{i(k_j'-k_j) L}-1}{N[e^{i(k_j'-k_j) L/N }-1]} & \text{otherwise}
\end{cases}
=
\begin{cases}
1 & \text{if } (k_j-k_j')\in \frac{2\pi N}{L} \mathbb{Z} \\
0 & \text{otherwise},
\end{cases}
\end{equation}
since $k_j, k_j'\in \frac{2\pi}{L}\mathbb{Z}$, so $e^{i(k_j'-k_j) L}=1$. Therefore, 
\begin{equation}
\widehat{\mathcal{P}f}(k) 
=
\sum_{k'\, |\, (k-k') \in \frac{2\pi N}{L} \mathbb{Z}^d} \widehat{f}\,(k') 
=
\sum_{k' \in \frac{2\pi}{L} \mathbb{Z}^d}  \widehat{f}\,(k + N k'). 
\end{equation}
\end{proof}

\begin{lemma}
Let $y \in \mathbb{R}^d$ be a vector with $\|y\|_\infty \le 1/2$ for $d \ge 1$, and let $q \ge d+2\pi$, then
\begin{equation}
\sum_{x\in \mathbb{Z}^d \setminus \{0\} } \frac{1}{ \| x+y \|^q}  \le \Big[ (1+\sqrt{d}) \sqrt{d+3} \Big]^q.
\end{equation}
\label{lemma:sum}
\end{lemma}
\begin{proof}
We will analyze the $y=0$ case first, and use the result to derive a bound for any $y$. For $x \in \mathbb{Z}^d \setminus \{0\}$, $\|x\| \ge 1$, so $\|x\|^2 \ge (\|x\|^2+1)/2 > 0$, and therefore \\
\begin{equation}
\frac{1}{\|x\|^q} \le \left( \frac{2}{\|x\|^2 + 1} \right)^{q/2}.
\end{equation}
Define $C_x = [-\frac{1}{2}, \frac{1}{2}]^d  \,+ \, x $, the unit cube centered around this $x$. For any $z=(z_1, \dots, z_d)^\top \in C_x$,
\begin{equation}
z_i^2 \le \Big(|x_i| + \frac{1}{2}\Big)^2 \le 2x_i^2 + \frac{1}{2},
\end{equation}
so
\begin{equation}
\|z\|^2+1 \le 2 \|x\|^2+ \frac{d}{2}+1 \le \Big(\frac{d+3}{2}\Big) \Big( \|x\|^2 + 1 \Big).
\label{eq:inf_sum_bound}
\end{equation}
Therefore, 
\begin{align}
\sum_{x\in \mathbb{Z}^d\setminus \{ 0\} } \frac{1}{\|x\|^q}
&\le 
\sum_{x\in \mathbb{Z}^d\setminus \{ 0\} } \left( \frac{2}{\|x\|^2 + 1} \right)^{q/2} \\ 
&= 
\sum_{x\in \mathbb{Z}^d\setminus \{ 0\} } \int_{C_x} \left( \frac{2}{\|x\|^2 + 1} \right)^{q/2} dz \\
&\le 
\sum_{x\in \mathbb{Z}^d\setminus \{ 0\} } \int_{C_x} \left( \frac{d+3}{\|z\|^2 + 1} \right)^{q/2} dz \\ 
& \le 
\sum_{x\in \mathbb{Z}^d } \int_{C_x} \left( \frac{d+3}{\|z\|^2 + 1} \right)^{q/2} dz \\ 
&= 
\int_{\mathbb{R}^d} \left( \frac{d+3}{\|z\|^2 + 1} \right)^{q/2} dz \\ 
&=
\frac{2 \pi^{d/2}}{\Gamma(d/2)} \int_0^\infty \left( \frac{d+3}{r^2+1} \right)^{q/2} r^{d-1}dr \\
&= 
(d+3)^{q/2} \, \pi^{d/2} \, \frac{\Gamma[(q-d)/2]}{\Gamma[q/2]} \qquad \qquad (\text{since }d<q),
\end{align}
where we used spherical coordinates to evaluate the integral over $\mathbb{R}^d$. Then using the identity \cite{NIST:DLMF} (5.6.8)
\begin{equation}
\left| \frac{\Gamma(w+a)}{\Gamma(w+b)} \right| \le |w|^{a-b} 
\qquad \qquad \text{for } \text{Re}(w)>0 \text{ and } 0 \le a \le b-1,
\end{equation}
with $w=(q-d)/2$, $a=0$, and $b=d/2 \ge 1$, gives
\begin{align}
\sum_{x\in \mathbb{Z}^d\setminus \{ 0\} } \frac{1}{\|x\|^q}
& \le 
(d+3)^{q/2}  \left(\frac{2\pi}{q-d} \right)^{d/2} \le (d+3)^{q/2} 
\end{align}
for $d\ge 2$, since $q \ge d + 2 \pi$ by assumption. We will now use this result to prove Eq.~\eqref{eq:inf_sum_bound} for $d \ge 2$, then show that the resulting bound also holds for $d=1$. For any $d\ge 1$ we have
\begin{equation}
\frac{1}{\| x+y \|} \le \frac{1+\sqrt{d}}{\sqrt{d}/2 + \|x+y \|} \le \frac{1+\sqrt{d}}{\|x\|},
\end{equation}
where the first step follows from $\|x+y\| \ge 1/2$, and the second from $\|y\| \le \sqrt{d} \|y\|_\infty \le \sqrt{d}/2$, therefore $\|x+y\| \ge \|x\| - \|y \| \ge \|x\| - \sqrt{d}/2$. For $d \ge 2$ this gives Eq.~\eqref{eq:inf_sum_bound}:
\begin{equation}
\sum_{x\in \mathbb{Z}^d \setminus \{0\} } \frac{1}{ \| x+y \|^q} 
\le 
(1+\sqrt{d})^q \sum_{x\in \mathbb{Z}^d \setminus \{0\} } \frac{1}{ \| x \|^q} 
\le 
\Big[ (1+\sqrt{d}) \sqrt{d+3} \Big]^q.
\end{equation}

To complete the proof, we show that the $d=1$ case has the same upper bound:
\begin{align}
\sum_{x \in \mathbb{Z}\setminus\{0\}}
\frac{1}{|x+y|^q}
&= 
\sum_{x=1}^\infty \frac{1}{|x+y|^q}
+ \sum_{x=1}^\infty \frac{1}{|x-y|^q} \\
& \le 
2 \sum_{x=1}^\infty \frac{1}{|x-1/2|^q} \\
& \le 
2 \left[ 2^q + \int_1^\infty \frac{dx}{(x-1/2)^q} \right] \\ 
&=
\frac{2^q(2q-1)}{q-1} \; \le \; 4^q.
\end{align}
\end{proof}

\begin{lemma}
Let $f \in \mathsf{C}^{2(m+s)}$ for integers $m \ge 0$ and $s \ge  (d+7)/4$, then:
\begin{equation}
\big \| \nabla^{2m} ( \mathcal{P} f-f) \big \|_{\mathsf{L}^2} \le \left( \frac{Ld}{N} \right)^{2s} \big\| \nabla^{2(m+s)} f \big\|_{\mathsf{L}^2}.
\end{equation}
\label{lemma:Pf-f}
\end{lemma}

\begin{proof}
The Fourier series of $\nabla^2 f$ is 
\begin{equation}
(\nabla^2 f)(x) = \sum_{k\in \frac{2\pi}{L} \mathbb{Z}^d} \big( - \|k\|^2 \big) \widehat{f}\, (k) \, e^{i k\cdot x},
\end{equation}
so by induction, that of $\nabla^{2m} f$ is
\begin{equation}
(\nabla^{2m} f)(x) = \sum_{k\in \frac{2\pi}{L} \mathbb{Z}^d} \big( - \|k\|^2 \big)^m \widehat{f}\, (k) \, e^{i k\cdot x}.
\end{equation}
Term-wise differentiation is justified since $f \in \mathsf{C}^{2(m+s)}$ for positive $s$. Similarly, the Fourier series of $\nabla^{2m} \mathcal{P}f$ is
\begin{equation}
\big (\nabla^{2m} \mathcal{P} f \big)(x) = \sum_{k\in \mathbb{K}_N^d} \big( - \|k\|^2 \big)^m \, \widehat{\mathcal{P} f}\, (k) \, e^{i k\cdot x},
\end{equation}
so that of $\nabla^{2m}(\mathcal{P}f-f)$ is
\begin{equation}
\big[ \nabla^{2m} (\mathcal{P}f-f) \big](x)
=
\sum_{k \in \mathbb{K}_N^d} \Big( - \|k\|^2 \Big)^m \Big[  \widehat{\mathcal{P}f}\,(k) - \widehat{f}(k) \Big] \, e^{i k\cdot x}
-
\sum_{k \in (\frac{2\pi}{L} \mathbb{Z}^d) \setminus \mathbb{K}_N^d} \Big( - \|k\|^2 \Big)^m  \widehat{f}(k) \, e^{i k\cdot x}.
\end{equation}
Since $\nabla^{2m} ( \mathcal{P} f-f)$ is continuous and therefore square-integrable, we can express its norm using Parseval's theorem:
\begin{equation}
\big \| \nabla^{2m} ( \mathcal{P} f-f) \big \|_{\mathsf{L}^2}^2
=
\underbrace{L^d \vphantom{\sum_{k \in (\frac{2\pi}{L} \mathbb{Z}^d) \setminus \mathbb{K}_N^d}} \sum_{k \in \mathbb{K}_N^d}  \|k\|^{4m} \; \Big|  \widehat{\mathcal{P}f}\,(k) - \widehat{f}(k) \Big|^2}_\text{Term 1 (aliasing)}
+
\underbrace{L^d \sum_{k \in (\frac{2\pi}{L} \mathbb{Z}^d) \setminus \mathbb{K}_N^d}  \|k\|^{4m} \; \big|  \widehat{f}(k) \big|^2}_\text{Term 2 (bandlimiting)},
\end{equation}
where the first sum describes aliasing (high spatial frequencies of $f$ being folded into low frequencies of $\mathcal{P}f$, following Eq.~\eqref{eq:aliasing}), and the second reflects the absence of high frequencies in $\mathcal{P}f$.\\

We will start by bounding the first term, using the aliasing formula in Lemma \ref{lemma:aliasing} to compare $\widehat{\mathcal{P}f}$ and $\widehat{f}$, then Lemma~\ref{lemma:sum} to bound the resulting sum. We are justified in using aliasing formula because $f\in \mathsf{C}^{2(m+s)}$ for $2(m+s) >d/2$.
\begin{align}
\text{Term 1} 
&=
L^d \sum_{k \in \mathbb{K}_N^d}  \|k\|^{4m} \bigg| \widehat{f}(k) - \sum_{k' \in \frac{2\pi}{L} \mathbb{Z}^d} \widehat{f} \, (k+N k') \bigg|^2 \\
&=
L^d \sum_{k \in \mathbb{K}_N^d} \|k\|^{4m}  \bigg| \sum_{k' \in \frac{2\pi}{L} \mathbb{Z}^d \setminus \{0\} } \widehat{f} \, (k+N k') \bigg|^2 \\
&=
L^d \sum_{k \in \mathbb{K}_N^d} \Bigg| \sum_{k' \in \frac{2\pi}{L} \mathbb{Z}^d \setminus \{0\} } \frac{\|k\|^{2m}}{\|k+N k'\|^{2(m+s)}} \, \|k+Nk'\|^{2(m+s)} \, \widehat{f} \, (k+N k') \Bigg|^2 \\
& \le 
L^d \sum_{k \in \mathbb{K}_N^d} \left[ \sum_{\ell \in \frac{2\pi}{L} \mathbb{Z}^d \setminus \{0\} } \frac{\|k\|^{4m}}{\|k +N \ell\|^{4(m+s)}} \right] 
\left[  \sum_{k' \in \frac{2\pi}{L} \mathbb{Z}^d \setminus \{0\} } \|k+Nk'\|^{4(m+s)} \, \big| \widehat{f} \, (k+N k') \big|^2 \right] \\
& \le 
L^d \sum_{k \in \mathbb{K}_N^d} \left[ \sum_{\ell \in \frac{2\pi}{L} \mathbb{Z}^d \setminus \{0\} } \frac{1}{\|k +N \ell\|^{4s}} \right] 
\left[  \sum_{k' \in \frac{2\pi}{L} \mathbb{Z}^d \setminus \{0\} } \|k+Nk'\|^{4(m+s)} \, \big| \widehat{f} \, (k+N k') \big|^2 \right] \\ 
& \le 
L^d \left(\frac{L}{2\pi N}\right)^{4s} \left[ \max_{k \in \mathbb{K}_N^d} \sum_{\ell \in \mathbb{Z}^d \setminus \{0\} } \frac{1}{\|\ell + \frac{k L}{2\pi N}\|^{4s}} \right] \sum_{k \in \mathbb{K}_N^d} \sum_{k' \in \frac{2\pi}{L} \mathbb{Z}^d \setminus \{0\} } \|k+Nk'\|^{4(m+s)} \, \big| \widehat{f} \, (k+N k') \big|^2  \\
& \le 
L^d \left(\frac{L}{2\pi N}\right)^{4s} \Big[ (1+\sqrt{d}) \sqrt{d+3} \Big]^{4s} \sum_{k \in \frac{2\pi}{L} \mathbb{Z}^d} \|k\|^{4(m+s)} \, \big| \widehat{f} \, (k) \big|^2 \\ 
& =
\left(\frac{L}{2\pi N}\right)^{4s} \Big[ (1+\sqrt{d}) \sqrt{d+3} \Big]^{4s} \, \big \| \nabla^{2(m+s)} f \big \|_{\mathsf{L}^2}^2
\end{align}
Here we have used Cauchy-Schwarz to split the initial sum over $k'$, then the fact that $\|\frac{k L}{2\pi N}\|_\infty \le \frac{1}{2}$ for all $k \in \mathbb{K}_N^d$ to apply Lemma~\ref{lemma:sum} with $q=4s$. The use of that lemma is justified since $4s \ge d+7 > d+2\pi$. Similarly, since $\nabla^{2(m+s)} f$ is continuous, applying Parseval's theorem and term-wise differentiation to it is justified.

The second term can be bounded similarly, with the same justification:
\begin{align}
\text{Term 2} &= 
L^d \sum_{k \in (\frac{2\pi}{L} \mathbb{Z}^d) \setminus \mathbb{K}_N^d}  \frac{1}{\|k\|^{4s}} \;  \|k\|^{4(m+s)} \; \big|  \widehat{f}(k) \big|^2 \\
& \le  
L^d \left( \frac{L}{\pi N \sqrt{d}} \right)^{4s} \sum_{k \in (\frac{2\pi}{L} \mathbb{Z}^d) \setminus \mathbb{K}_N^d}  \|k\|^{4(m+s)} \; \big|  \widehat{f}(k) \big|^2 \\ 
& \le 
L^d \left( \frac{L}{\pi N \sqrt{d}} \right)^{4s} \sum_{k \in \frac{2\pi}{L} \mathbb{Z}^d} \|k\|^{4(m+s)} \; \big|  \widehat{f}(k) \big|^2 \\
&=
\left( \frac{L}{\pi N \sqrt{d}} \right)^{4s} \big \| \nabla^{2(m+s)} f \big \|_{\mathsf{L}^2}^2.
\end{align}
Combining the bounds for both terms gives the desired result:
\begin{align}
\big \| \nabla^{2m} ( \mathcal{P} f-f) \big \|_{\mathsf{L}^2}^2 
& \le 
\left[ \frac{L}{\pi N} \left( \frac{ (1+\sqrt{d}) \sqrt{d+3}}{2} + \frac{1}{\sqrt{d}} \right) \right]^{4s}  \big \| \nabla^{2(m+s)} f \big \|_{\mathsf{L}^2}^2 \\ 
& \le 
 \left( \frac{Ld}{ N} \right)^{4s} \big \| \nabla^{2(m+s)} f \big \|_{\mathsf{L}^2}^2.
\end{align}
\end{proof}

\begin{corollary}
Under the same conditions as Lemma~\ref{lemma:Pf-f}:
\begin{equation}
\big\| \big [\nabla^{2m} , \mathcal{P} \big ] f \big\|_{\mathsf{L}^2}  \le 2 \left( \frac{Ld}{N} \right)^{2s} \big\| \nabla^{2(m+s)} f \big\|_{\mathsf{L}^2}.
\end{equation}
\label{cor:commutator_norm}
\end{corollary}

\begin{proof}
\begin{align}
\big\| \big [\nabla^{2m} , \mathcal{P} \big ] f \big\|_{\mathsf{L}^2}
&=
\big\| \nabla^{2m}(\mathcal{P}f-f) - (\mathcal{P} \nabla^{2m} f - \nabla^{2m}f ) \big\|_{\mathsf{L}^2} \\ 
& \le
\big\| \nabla^{2m}(\mathcal{P}f-f) \big\|_{\mathsf{L}^2}
+
\big\| \mathcal{P} (\nabla^{2m}f) - (\nabla^{2m}f) \big\|_{\mathsf{L}^2}\\ 
& \le 
2 \left( \frac{Ld}{N} \right)^{2s} \big\| \nabla^{2(m+s)} f \big\|_{\mathsf{L}^2}.
\end{align}
\end{proof}

\subsection{Bounding (finite-dimensional) vector differential equations}

\begin{lemma}

Suppose $z_t\in \mathbb{C}^D$ satisfies $i \frac{d}{dt} z_t = H_t \, z_t + b_t$, where $H_t$ is a $D\times D$ Hermitian matrix and $b_t \in \mathbb{C}^D$, for all $t\in[0,T]$. If the elements of $H_t$ and $b_t$ are all continuous function of $t$, then 
\begin{equation}
\| z_T \| \le \| z_0 \| + \int_0^T \|b_t\| \, dt.
\end{equation}
\label{lemma:DE_bound}
\end{lemma}

\begin{proof}
Let $U_t = \mathcal{T} \exp(-i \int_0^t H_\tau \, d\tau )$ be the unique matrix-valued function satisfying $i \frac{d}{dt} U_t = H_t U_t$ and $U_0=I$, then $U_t$ is unitary and its elements are continuously differentiable functions of $t$. Define $y_t = U_t^\dag z_t$, then $\|y_t\| = \|z_t\|$, and $y_t$ satisfies the differential equation (DE)
\begin{equation}
i\frac{d}{dt} \, y_t
=
\Big (-i \frac{d}{dt} U_t \Big)^\dag z_t + U_t^\dag \Big(i \frac{d}{dt} z_t \Big)
=
(-U_t^\dag H_t) z_t + U_t^\dag (H_t z_t + b_t)
=
U_t^\dag b_t,
\label{eq:DE_for_y}
\end{equation}
so $\| \frac{d}{dt} \, y_t \| = \|b_t \|$. Moreover, since the elements of $U_t$ and $b_t$ are continuous in $t$, so are those of $\frac{d}{dt} y_t$, so we can use the fundamental theorem of calculus to write
\begin{equation}
y_T - y_0 = \int_0^T \frac{d}{dt} \, y_t \, dt.
\end{equation}
Therefore, we can use the reverse triangle inequality to get the desired result:
\begin{align}
\|z_T\| -  \|z_0\| 
=
\|y_T\| -  \|y_0\| 
\le 
\|y_T - y_0\| 
=
\Big \| \int_0^T \frac{d}{dt} \, y_t \, dt \Big \|
\le 
\int_0^T \big \| \frac{d}{dt} \, y_t \big \| \, dt
=
\int_0^T \big \| b_t \big \| \, dt.
\end{align}
\end{proof}

\begin{proof}[Proof of Theorem~\ref{thm:spatial_disc_error}]

Define the real \textbf{normalization factor} $a_t >0$ such that
\begin{equation}
\ket{\psi_t}
=
a_t \left( \frac{L}{N} \right)^{d/2} \sum_{x\in \mathbb{X}_N^d} \Psi_t(x) \ket{x}
\end{equation}
is a unit vector, where $\Psi_t(x) = \sqrt{p_t(x)}$ is the true wavefunction evaluated at the grid point $x \in \mathbb{X}_N^d$. Since $(L/N)^d \sum_{x\in \mathbb{X}_N^d} |\Psi_t(x)|^2$ is a Riemann sum for $\int_{\mathbb{T}^d} p_t(x) dx = 1$, this $a_t$ will be close to 1 for large $N$. Specifically, 
\begin{equation}
a_t^{-1}
=
\left( \frac{L}{N} \right)^{d/2} \sqrt{\sum_{x\in \mathbb{X}_N^d} \big| \Psi_t(x) \big|^2 }
=
\left( \frac{L}{N} \right)^{d/2} \sqrt{\sum_{x\in \mathbb{X}_N^d} \big| (\mathcal{P}\Psi_t)(x) \big|^2 }
= 
\big \| \mathcal{P} \Psi_t \big \|_{\mathsf{L}^2}
\label{eq:at_inv}
\end{equation}
by Lemma \ref{lemma:parseval_X}, since $\mathcal{P} \Psi_t \in \mathsf{B}$ and $(\mathcal{P}\Psi_t)(x) = \Psi_t(x)$ for all $x \in \mathbb{X}_N^d$. Clearly $a_t^{-1} < \infty$. We can also lower-bound it using the reverse triangle inequality and Lemma~\ref{lemma:Pf-f}, since $\Psi_t \in \mathsf{C}^{2(s+1)}$:
\begin{equation}
a_t^{-1} = \big \|\mathcal{P}\Psi_t - \Psi_t + \Psi_t \big \|_{\mathsf{L}^2}
\ge 
1 -  \big \|\mathcal{P}\Psi_t - \Psi_t \big \|_{\mathsf{L}^2}
\ge 
1 - \left( \frac{L d}{N} \right)^{2(s+1)} \big \| \nabla^{2(s+1)} \Psi_t \big \|_{\mathsf{L}^2}.
\end{equation}
Note that $\delta \le T$ implies $c_s (Ld/N)^{2s} \le 1$, and therefore $Ld \le N$ since $c_s \ge 1$. Moreover, $\big \| \nabla^{2(s+1)} \Psi_t \big \|_{\mathsf{L}^2} \le c_s/3$, therefore:
\begin{equation}
a_t^{-1} 
\ge 
1 - \left( \frac{L d}{N} \right)^{2s} \big \| \nabla^{2(s+1)} \Psi_t \big \|_{\mathsf{L}^2}
\ge 
1 - \frac{c_s}{3} \left( \frac{L d}{N} \right)^{2s} 
\ge \frac{2}{3},
\label{eq:a_bound}
\end{equation}
so $0 < a_t \le 3/2$ for any $t \in [0,T]$. We will use this result in the rest of the proof.\\

To bound the distance between the target state $\ket{\psi_T}$ and the state $\ket{\phi_T}$ generated by $H_t^c$, we will use the triangle inequality to split the spatial discretization error into two terms:
\begin{equation}
\big \| \ket{\psi_T} - \ket{\phi_T} \big \|
\le 
\underbrace{\Big \|\ket{\phi_T} - \frac{a_0}{a_T} \, \ket{\psi_T} \Big \|}_\text{Term 1 (drift)}
+
\underbrace{ \Big \| \Big( \frac{a_0}{a_T}-1 \Big) \, \ket{\psi_T}\Big \|}_\text{Term 2 (normalization)}.
\label{eq:spatial_triangle_ineq}
\end{equation}
The first term will describe how $\ket{\phi_t}$ drifts away from an unnormalized version of the target state, while the second term will describe how the normalization factor $a_t$ in the target state changes over time.\\ 

To bound \textbf{Term 1}, we will apply the projector $\mathcal{P}$ to both sides of the continuous-space Schr\"odinger equation $i\frac{\partial}{\partial t} \Psi_t = \hat{\mathcal{H}}_t^c \Psi_t$ to get a DE where both sides are contained in $\mathsf{B}$:
\begin{align}
i \mathcal{P} \frac{\partial}{\partial t} \Psi_t
=
i \frac{\partial}{\partial t} \mathcal{P} \Psi_t
&= 
\frac{i}{2} \mathcal{P} \Big[ V_t \nabla^2 \Psi_t - \nabla^2 (V_t \, \Psi_t) \Big]\\
&=
\frac{i}{2} \Big\{ \mathcal{P}(V_t \nabla^2 \mathcal{P} \Psi_t)  - \nabla^2 \mathcal{P} (V_t \, \Psi_t) \Big\}
+
\frac{i}{2} \Big \{ \mathcal{P} \big( V_t  [\mathcal{P}, \nabla^2] \Psi_t \big) - [\mathcal{P}, \nabla^2](V_t \, \Psi_t) \Big \}. 
\label{eq:projected_schrodinger}
\end{align}
The first step, namely $\mathcal{P} \frac{\partial}{\partial t} \Psi_t = \frac{\partial}{\partial t} \mathcal{P} \Psi_t$, follows immediately from the definition of $\mathcal{P}$ in Eq.~\eqref{eq:P_def}. Likewise, to get the second line, we added and subtracted terms using the fact that for any function $f$, its projection $\mathcal{P}f \in \mathsf{B}$ is fully determined by the values $f(x)$ at grid points $x\in \mathbb{X}_N^d$, and $f(x)=(\mathcal{P}f)(x)$ for all $x\in \mathbb{X}_N^d$. Therefore for any functions $f$ and $g$, $\mathcal{P}(fg) = \mathcal{P}[(\mathcal{P}f)g] = \mathcal{P}[f(\mathcal{P}g)]$.\\

The point of writing the equation this way is that the first term in \eqref{eq:projected_schrodinger} involves only Laplacians of functions in $\mathsf{B}$, for which we have an exact expression (from Lemma~\ref{lemma:laplacian}), while the second term involves $[\mathcal{P}, \nabla^2]$, which is small for large $N$ (in the sense of Corollary~\ref{cor:commutator_norm}).\\

Evaluating the previous equation at a grid point gives the following DE (in the time parameter $t$):
\begin{align}
i \frac{\partial}{\partial t} \Psi_t(x)
&=
-i V_t(x) \;\bra{x} K \! \left [\; \sum_{x'\in \mathbb{X}_N^d} \Psi_t(x') \ket{x'} \right]
+ i \bra{x} K \! \left[ \, \sum_{x'\in \mathbb{X}_N^d} V_t(x') \Psi_t(x') \ket{x'}  \right] \\
& \qquad \qquad \qquad \qquad \; \,
+\frac{i}{2} \Big \{ \mathcal{P} \big( V_t  [\mathcal{P}, \nabla^2] \Psi_t \big)(x) - [\mathcal{P}, \nabla^2](V_t \, \Psi_t)(x) \Big \} \nonumber \\[1ex]
&=
-i \bra{x} \Big( D_{V_t} K - K D_{V_t} \Big) \sum_{x'\in \mathbb{X}_N^d} \Psi_t(x') \ket{x'}
+ \frac{i}{2} \Big \{ \mathcal{P} \big( V_t  [\mathcal{P}, \nabla^2] \Psi_t \big)(x) - [\mathcal{P}, \nabla^2](V_t \, \Psi_t)(x) \Big \},
\end{align}
for any $x \in \mathbb{X}_N^d$. Equivalently, we can combine the DEs for each such $x$, and multiply by $a_0 (L/N)^{d/2}$, to get a single DE for the $N^d$-dimensional vector $\frac{a_0}{a_t} \ket{\psi_t} =  (L/N)^{d/2} \sum_{x\in \mathbb{X}_N^d}\Psi_t(x) \ket{x}$:
\begin{equation}
i \frac{d}{dt} \Big( \frac{a_0}{a_t} \ket{\psi_t} \Big) 
=
H_t^c \Big( \frac{a_0}{a_t} \ket{\psi_t} \Big) 
+ \underbrace{
\frac{i a_0}{2} \left( \frac{L}{N} \right)^{d/2}
\sum_{x\in \mathbb{X}_N^d} \Big \{ \mathcal{P} \big( V_t  [\mathcal{P}, \nabla^2] \Psi_t \big)(x) - [\mathcal{P}, \nabla^2](V_t \, \Psi_t)(x) \Big \} \ket{x}
}_{b_t \; \in \; \mathbb{C}^{N^d}}.
\label{eq:schrodinger+b}
\end{equation}
Notice that this is not a Schr\"odinger equation because of the $b_t$ term, which reflects the fact that the original continuous-space Schr\"odinger equation $i \frac{\partial}{\partial t} \Psi_t = \hat{\mathcal{H}}_t^c \Psi_t$ involves components both inside and outside of $\mathsf{B}$, which are coupled by $\hat{\mathcal{H}}_t^c$. Therefore, the unitary dynamics it generates appear non-unitary when projected onto $\mathsf{B}$. Note, however, that since $\ket{\phi_t}$ obeys the finite-dimensional Schr\"odinger equation $i \frac{d}{d t} \ket{\phi_t} = H_t^c \ket{\phi_t}$ by definition, $z_t = \ket{\phi_t} - \frac{a_0}{a_t} \ket{\psi_t}$ obeys
\begin{equation}
i \frac{d}{dt} z_t = H_t^c z_t - b_t
\end{equation}
with the initial condition $z_0 = \ket{\phi_0} - \ket{\psi_0}$. We do not know much about $b_t$, but we can use Lemma~\ref{lemma:parseval_X} to express its Euclidean norm in terms of an $\mathsf{L}^2$ norm, which we can then bound using Corollary~\ref{cor:commutator_norm}:
\begin{align}
\big \|b_t \big \| 
& \stackrel{(\text{Lemma}~\ref{lemma:parseval_X})}{=} \frac{a_0}{2} \; \Big \| \mathcal{P} \big( V_t  [\mathcal{P}, \nabla^2] \Psi_t \big) - [\mathcal{P}, \nabla^2](V_t \, \Psi_t) \Big\|_{\mathsf{L}^2} \\
&  
\, \stackrel{(\text{Cor.}~\ref{cor:commutator_norm})}{\le} \;\,
a_0 \left( \frac{Ld}{N} \right)^{2s} \bigg[ \big \| V_t \big \|_{\mathsf{L}^\infty} \, \big \|\nabla^{2(s+1)}  \Psi_t \big \|_{\mathsf{L}^2}
+
\big \|\nabla^{2(s+1)} (V_t \, \Psi_t) \big \|_{\mathsf{L}^2}
\bigg] \\ 
& \quad \; \; \; \le 
 \quad \; \; \frac{c_s}{2} \left( \frac{Ld}{N} \right)^{2s},
\end{align}
since the term in square brackets is less than $c_s/3$, and $a_0 \le 3/2$ from Eq.~\eqref{eq:a_bound}. The use of Corollary~\ref{cor:commutator_norm} is justified since $\Psi_t, \; V_t \Psi_t \in \mathsf{C}^{2(s+1)}$. Moreover, since $V_t(x)$ is continuous in $t$, the elements of $H_t^c$ are too. Likewise, since $\nabla^2 \Psi_t(x)$ and $\nabla^2 [V_t(x) \Psi_t(x)]$ are also continuous in $t$, the elements of $b_t$ are too. We can therefore use Lemma~\ref{lemma:DE_bound} to get
\begin{align}
\text{Term 1} 
\;\; = \;\;   
\big \|z_T \big \| \le \big\| \ket{\phi_0} - \ket{\psi_0} \big \| + \int_0^T \big \|b_t \big \| dt
\;\; \le \;\; 
\big\| \ket{\phi_0} - \ket{\psi_0} \big \| + \frac{T c_s}{2} \left( \frac{Ld}{N} \right)^{2s}.
\end{align}

To bound \textbf{Term 2}, following the proof of Lemma~\ref{lemma:DE_bound}, let $U_t$ be the unique matrix-valued function satisfying $i \frac{d}{dt} U_t = H_t^c U_t$ and $U_0=I$, so that $\ket{\phi_t}=U_t \ket{\phi_0}$. Let $y_t = \frac{a_0}{a_t} U_t^\dag \ket{\psi_t}$, then $\|y_t \| = \frac{a_0}{a_t}$ and
\begin{equation}
i \frac{d}{dt} y_t 
= 
\Big(-i \frac{d}{dt} U_t \Big)^\dag \, \Big(\frac{a_0}{a_t} \ket{\psi_t}\Big) + i U_t^\dag \frac{d}{dt} \Big( \frac{a_0}{a_t} \ket{\psi_t} \Big)
=
-U_t^\dag  H_t^c \Big(\frac{a_0}{a_t} \ket{\psi_t}\Big) + U_t^\dag \bigg[ H_t^c \Big(\frac{a_0}{a_t} \ket{\psi_t}\Big) +  b_t\bigg]
= U_t^\dag b_t
\end{equation}
using Eq.~\eqref{eq:schrodinger+b}. Note that $a_0/a_t$ is continuously differentiable since it depends only on $\Psi_t(x)$ at grid points $x\in \mathbb{X}_N^d$ as given by Eq.~\eqref{eq:at_inv}, and $\hat{\mathcal{H}}_t^c \Psi_t(x)$ is continuous so $i \frac{\partial}{\partial t} \Psi_t(x) = \hat{\mathcal{H}}^c_t \Psi_t(x)$ is too. We can use the previous equation to get
\begin{equation}
\Big| \frac{d}{dt} \, \Big(\frac{a_0}{a_t}\Big) \Big|
=
\Big| \frac{d}{dt} \|y_t \| \Big | \le \Big\| \frac{d}{dt} y_t \Big\| = \|b_t \|,
\end{equation}
where we used the reverse triangle inequality to bring the derivative inside the norm. Finally, using the continuity of $\frac{d}{dt} (a_0/a_t)$, we can use the fundamental theorem of calculus to write
\begin{equation}
\text{Term 2} 
=
\left| \int_0^T \frac{d}{dt} \left( \frac{a_0}{a_t} \right)  dt \right|
\le 
\int_0^T  \left| \frac{d}{dt} \left( \frac{a_0}{a_t} \right) \right| dt
\le 
\int_0^T \|b_t \| dt
\le \frac{T c_s}{2} \left( \frac{L d}{N} \right)^{2s}.
\end{equation}

\end{proof}

\section{Hamiltonian simulation details---discretizing time}
\label{app:discretizing_time}

We proceed in two steps: first we approximate $U(t_1,t_0)$ by the exponential of a constant Hamiltonian in Lemma~\ref{lemma:time_dep_approx}, then we approximate the latter unitary by a product formula that can be efficiently implemented in Lemma~\ref{lemma:PF_error}. Combining both results, plus rescaling to minimize the aggregate error bound, immediately proves Theorem~\ref{thm:time_disc_error}.

The simplest way to approximate evolution by a time-dependent Hamiltonian $H(t)$ over an interval $[t_0,t_1]$ is with the time-independent Hamiltonian $H(t_0)$. Note that a simple, randomized variant of this approach \cite{poulin:illusion} gives a similar bound that depends on $\|H(t)\|$ rather than $\|H'(t)\|$, which would simplify the overall error bound, at the cost of introducing mixed states and the diamond norm. We therefore leave this and similar known improvements \cite{berry:2020, an:2022} as subjects for future work.

\begin{lemma} 
\label{lemma:time_dep_approx}
Let $H(t)$ be a finite-dimensional Hamiltonian that is continuously differentiable in $t$, and let $U(t_1, t_0) = \mathcal{T} \exp[-i \int_{t_0}^{t_1} H(t) \, dt]$ for $t_1 \ge t_0$ be the unitary generated by $H(t)$ over the time interval $[t_0, t_1]$. Then $U(t_1, t_0)$ can be approximated by the unitary $e^{-i H(t_0) \Delta t}$, which describes evolution by the time-independent Hamiltonian $H(t_0)$ for time $\Delta t = t_1 - t_0$, with a spectral norm error of
\begin{equation}
\big\| U(t_1, t_0) - e^{-i H(t_0) \Delta t} \big \|
\le 
\frac{\Delta t^2}{2} \max_{t \in [t_0, t_1]} \big \| H'(t) \big\|.
\end{equation}
\end{lemma}
\begin{proof}
Let $A(t) = U(t_0 + t, \, t_0)$ and $B(t)=e^{-iH(t_0) t}$, then $i A'(t)= H(t_0+t) A(t)$ and $i B'(t) = H(t_0) B(t)$, so
\begin{align}
\big\|  U(t_1, t_0) - e^{-i H(t_0) \Delta t} \big \|
&=
\big\| A(\Delta t)^\dagger B(\Delta t) - I\big \| \\ 
&= 
\bigg \| \int_0^{\Delta t} \Big[ A'(t)^\dagger B(t) + A(s)^\dagger B'(t) \Big] dt \bigg \| \\ 
&= 
\bigg \| \, i \int_0^{\Delta t} A(t)^\dagger \big[ H(t_0+t) - H(t_0) \big] B(t) \, dt \bigg \| \\ 
&= 
\bigg \| \, i \int_0^{\Delta t} \int_{t_0}^{t_0+t}  A(t)^\dagger H'(s) B(t) \, ds \, dt \bigg \| \\ 
& \le 
\max_{\tau \in [t_0, t_1]} \big \| H'(\tau) \big\| \int_0^{\Delta t} \int_{t_0}^{t_0+t} ds \, dt,
\end{align}
where we used the fundamental theorem of calculus to get from the first to the second line, and from the third to the forth line. 
\end{proof}

Next, we bound the error from simulating a constant Hamiltonian, given as a commutator, using a product formula. This is somewhat different than the usual setting where a Hamiltonian is given as a sum of Hermitian or unitary terms---here instead, it is expressed as a difference of terms that are neither. The simplest appropriate product formula, from Ref.~\cite{childs_wiebe_comm}, would incur an error of $O(\Delta t^{3/2})$ for a timestep $\Delta t$. Instead, we use the next simplest one, which gives an error of $O(\Delta t^2)$ like Lemma~\ref{lemma:time_dep_approx}, thereby making the two nicely compatible. Note that it is essential here to find exact error bounds, rather than just asymptotic expressions in $\Delta t$, since the prefactors will not be constant, but will instead depend on various problem parameters (e.g., $T$, $d$ etc.) because this Hamiltonian comes from a spatial discretization. We give a crude error bound in the following lemma, which we expect could be slightly improved by adapting the proof techniques from Refs.~\cite{childs:trotter_error_comm, Gluza2024}.

\begin{lemma} \label{lemma:PF_error}
Let $H=i[A,B]$ for Hermitian matrices $A$ and $B$, and let $S(t) = e^{itB} e^{itA} e^{-itB} e^{-itA}$ for $t\ge 0$, then
\begin{equation}
\Big \| S \big (\sqrt{\Delta t/2} \big) S \big(-\sqrt{\Delta t/2} \big) - e^{-iH  \Delta t} \Big\|
\le \left[ \frac{8}{3} \big(\|A\|+\|B\| \big)^4 + \frac{1}{2} \big \| H \big \|^2 \right] \Delta t^2.
\end{equation}
\end{lemma}

\begin{proof}
Let $\tau=\sqrt{\Delta t}$, then $F(\tau) = S \big(\tau/\sqrt{2} \big)S \big(-\tau/\sqrt{2} \big)$ is a product of matrix exponentials, so its elements are all smooth functions of $\tau$. We can therefore write it as a finite Taylor series about $\tau=0$ with an integral remainder term of order $O(\tau^4)$:
\begin{equation}
F(\tau) = I - i \tau^2 H +  \int_0^\tau \frac{(\tau-s)^3}{3!} \frac{d^4}{ds^4} F(s) \, ds.
\end{equation}
Likewise, the elements of $e^{-i H \Delta t}$ are smooth functions of $\Delta t$, so we have the Taylor series:
\begin{align}
e^{-i H \Delta t}
&=
I - i  H \Delta t + \int_0^{\Delta t} (\Delta t-s) \frac{d^2}{ds^2} e^{-i s H} ds \\ 
&= 
I - i \tau^2 H - \int_0^{\tau^2} (\tau^2 -s) H^2 e^{-i s H} ds.
\end{align}
Therefore, 
\begin{align}
\big \| F(\tau) - e^{-i \tau^2 H} \big \|
& \le 
\big \| F(\tau) - (I - i \tau^2 H) \big \|
+
\big \| e^{-i \tau^2 H} - (I - i \tau^2 H) \big \| \\ 
&=
\left \| \int_0^\tau \frac{(\tau-s)^3}{3!} \frac{d^4}{ds^4} F(s) \, ds \right\|
+
\left\| \int_0^{\tau^2} (\tau^2 -s) H^2 e^{-i s H} ds \right\| \\
& \le 
\max_{s \in[0,\tau]} \left\| \frac{d^4}{ds^4} F(s) \right\| \frac{\Delta t^2}{4!}
+
 \frac{\|H^2\| \Delta t^2}{2}.
\end{align}
We can get the claimed result by bounding $\|H^2\| \le \|H\|^2$, and using the general Leibniz rule to write the fourth derivative of $F$ in terms of a multi-index $c=(c_1, \dots, c_8)$ as:
\begin{equation}
\frac{d^4}{ds^4} F(s)
=
\sum_{\substack{c_1, \dots, c_8 \ge 0 \\ \|c\|_1 = 4}}
\frac{4!}{c_1! \cdots c_8!} 
\left( \frac{d^{c_1}}{ds^{c_1} } \; e^{i s B /\sqrt{2}} \right)
\left( \frac{d^{c_2}}{ds^{c_2} } \; e^{i s A /\sqrt{2}} \right)
\cdots 
\left( \frac{d^{c_8}}{ds^{c_8} } \; e^{i s A /\sqrt{2}} \right),
\end{equation}
so
\begin{align}
\left\| \frac{d^4}{ds^4} F(s) \right\|
&\le 
\sum_{\substack{c_1, \dots, c_8 \ge 0 \\ \|c\|_1 = 4}}
\frac{4!}{c_1! \cdots c_8!} 
\left \| \frac{d^{c_1}}{ds^{c_1} } \; e^{i s B /\sqrt{2}} \right \|
\left \| \frac{d^{c_2}}{ds^{c_2} } \; e^{i s A /\sqrt{2}} \right \|
\cdots 
\left \| \frac{d^{c_8}}{ds^{c_8} } \; e^{i s A /\sqrt{2}} \right \| \\ 
& \le
\sum_{\substack{c_1, \dots, c_8 \ge 0 \\ \|c\|_1 = 4}}
\frac{4!}{c_1! \cdots c_8!} 
\left \| B/\sqrt{2} \right \|^{c_1}
\left \| A/\sqrt{2} \right \|^{c_2} 
\cdots \;
\left \| A/\sqrt{2} \right \|^{c_8} \\ 
& = 
\Big( 4 \|A\|/\sqrt{2} +  4 \|B\|/\sqrt{2} \Big)^4,
\end{align}
using the multinomial theorem in the last step.
\end{proof}

\begin{proof}[Proof of Theorem~\ref{thm:time_disc_error}]
First, we simplify the error bound from Lemma~\ref{lemma:PF_error} as
\begin{equation}
\left[ \frac{8}{3} \big(\|A\|+\|B\| \big)^4 + \frac{1}{2} \big \| H \big \|^2 \right] \Delta t^2 
\le 
\left[ \frac{8}{3} \big(\|A\|+\|B\| \big)^4 + 2 \|A \|^2 \|B \|^2 \right] \Delta t^2
\le 
3 \big(\|A\|+\|B\| \big)^4 \Delta t^2.
\end{equation}
Naively, one could identify $A$ and $B$ with $K$ and $D_{V_{t_0}}$, respectively, but choosing $A=\gamma^{-1} K$ and $B=\gamma D_{V_{t_0}}$ instead for some $\gamma$ that makes $\|A\|$ and $\|B\|$ more similar can reduce the bound above without changing $i[A,B]$. To that end, we pick $\gamma= \|K\|^{1/2} = \frac{N\pi}{L} \sqrt{\frac{d}{2}}$ so that $\|A\|, \|B\| = O(N \sqrt{d}/L)$. (We could do something similar with $\|D_{V_t}\|$ too, but we opt not to here since that norm will not be known, in general, and should therefore not figure in the algorithm.) With this choice, the product formula in Lemma~\ref{lemma:PF_error} becomes that in Eqs.~\eqref{eq:PF} and \eqref{eq:PF_angles}. Invoking Lemmas~\ref{lemma:time_dep_approx} and \ref{lemma:PF_error} with the notation of Sec.~\ref{sec:discretizing_time}, and noting that $\|D_{V_t}\|= \|V_t\|_{\mathsf{L}^\infty}$ and $\|\frac{d}{dt} D_{V_t}\|= \|\frac{\partial}{\partial t} V_t\|_{\mathsf{L}^\infty}$ gives the desired result:
\begin{align}
\Big \| U(t_1, t_0) - W(t_0) \Big \|
& \le 
\Big \| W(t_0) - e^{-i H_{t_0}^c \Delta t} \Big \|
+ \Big \| U(t_1, t_0) - e^{-i H_{t_0}^c \Delta t} \Big \| \\ 
& \le 
3 \big( \|A\| + \|B\| \big)^4 \Delta t^2 + \frac{1}{2} \max_{t\in[t_0, t_1]} \Big \| \frac{d}{dt} H_t^c \Big\| \Delta t^2 \\ 
& \le 3 \left( \frac{\pi N}{L} \sqrt{\frac{d}{2}} +  \frac{\pi N}{L} \sqrt{\frac{d}{2}} \; \|V_{t_0}\|_{\mathsf{L}^\infty}  \right)^4 \Delta t^2
+
\|K \| \max_{t\in[t_0, t_1]} \Big \| \frac{d}{dt} D_{V_t} \Big\| \Delta t^2 \\ 
& = 
\frac{3\pi^4}{4} \frac{d^2 N^4}{L^4} \Big(1 + \|V_{t_0}\|_{\mathsf{L}^\infty}  \Big)^4 \Delta t^2
+ \frac{\pi^2}{2} \frac{d N^2}{L^2}  \max_{t\in[t_0, t_1]} \Big \| \frac{\partial}{\partial t} V_t \Big\|_{\mathsf{L}^\infty} \Delta t^2
\end{align}

\end{proof}

\section{Proof of Corollary~\ref{c:qmeansapprox}}\label{app:proof-cor-mean-estimation}

\textit{Proof of Corollary~\ref{c:qmeansapprox}}. We start by setting some notation. Define $h := L/N$ and $[N]^d := \{0,\ldots, N-1\}^d$, and denote $\|f\|_{\mathsf{L}^\infty} = M$. For any $k = (k_1,\ldots,k_d)\in [N]^d$ define $v_k = (k_1h,\ldots,k_dh)$, so that $\mathbb{X}^d_N = \{v_k\,|\,k\in [N]^d\}$. Additionally, recall that $\bar{p}_N$ is defined via \begin{equation}
    \bar{p}_N(x) = \frac{p(x)}{S}
\end{equation}
for all $x\in \mathbb{X}^d_N$, with
\begin{equation}
S = \sum_{k\in[N]^d}p(v_k) = \sum_{x\in\mathbb{X}^d_N}p(x).
\end{equation}
Recall also from Definition~\ref{def:app-disc-qsample-prep-access} that under the assumption of $\epsilon$-approximate discretized qsample preparation access with discretization scale $N$, we are given the ability to implement $U$ and $U^{-1}$ for the unitary satisfying $\|U|0\rangle - |\psi_{\bar{p}_N} \rangle \| \leq \epsilon$. Let us denote with $q:\mathbb{X}^d_N\rightarrow [0,1]$ the discrete probability distribution over $\mathbb{X}^d_N$ defined via
\begin{equation}
q(x) = |\langle x|U|0\rangle|^2 \text{ for all } x\in\mathbb{X}^d_N.
\end{equation}
Note that $q$ is the distribution sampled from when measuring $U|0\rangle$ in the computational basis, and it follows from $\| U|0\rangle -|\psi_{\bar{p}_N}\rangle\| \leq \epsilon$ that $\mathrm{d}_\mathrm{TV}(\bar{p}_N,q)\leq \epsilon$, where $\mathrm{d}_\mathrm{TV}$ denotes the total variation distance. Additionally, for any distribution $Q$ define 
\begin{align}
\mu(Q) &= \mathbb{E}_{x\sim Q}[f(x)], \\
m_2(Q) &= \mathbb{E}_{x\sim Q}[f^2(x)], \\
\sigma^2(Q) &= m_2(Q) - [\mu(Q)]^2.
\end{align}
With this in hand, we state the following helpful technical lemmas, whose proofs are given in Section~\ref{ss:lemma_proof}:

\begin{lemma}\label{lem:mean-variance-bound} Under the assumption that the PDF $p:[0,L]^d\rightarrow [0,\infty)$ is Lipschitz continuous with Lipschitz constant $\ell_p$, and that $f:[0,L]^d\rightarrow \mathbb{R}$ is Lipschitz continuous with Lipschitz constant $\ell_f$, we have that
\begin{align}
&|\mu(p) - \mu(\bar{p}_N)|\leq 2\|f\|_{\mathsf{L}^\infty}\frac{\ell_p\sqrt{d}L^{d+1}}{N} + \ell_f\sqrt{d}\frac{L}{N} \label{eq:mean_bound}\\
& |\sigma^2(p) - \sigma^2(\bar{p}_N)| \leq 6\|f\|^2_{\mathsf{L}^\infty}\frac{\ell_p\sqrt{d}L^{d+1}}{N} + 4 \|f\|_{\mathsf{L}^\infty} \ell_f\sqrt{d}\frac{L}{N} \label{eq:variance_bound}
\end{align}
\end{lemma}
\begin{lemma}\label{lem:TV-consequences} For any two discrete distributions $P,Q$ satisfying $\mathrm{d}_\mathrm{TV}(P,Q)\leq\epsilon$ one has
\begin{align}
&|\mu(P) - \mu(Q)| \leq 2\|f\|_{\infty}\epsilon \label{eq:mean-tv-bound}\\
&|\sigma^2(P) - \sigma^2(Q)| \leq 6\|f\|^2_{\infty}\epsilon\label{eq:var-tv-bound}
\end{align}
\end{lemma}

Using these lemmas, we then have
\begin{align}
|\mu(p) - \mu(q)| &\leq |\mu(p) - \mu(\bar{p}_N)| + |\mu(\bar{p}_N)- \mu(q)| \\
&\leq  2\|f\|_{\mathsf{L}^\infty}\frac{\ell_p\sqrt{d}L^{d+1}}{N} + \ell_f\sqrt{d}\frac{L}{N}  + 2\|f\|_{\infty}\epsilon,\\
&:= \epsilon_\mathrm{mean}\label{eq:mean-triangle}
\end{align}
where we have used both Eq.~\eqref{eq:mean_bound} from Lemma~\ref{lem:mean-variance-bound} and Eq.~\eqref{eq:mean-tv-bound} from Lemma~\ref{lem:TV-consequences}. Similarly, we also have  
\begin{align}
|\sigma^2(p) - \sigma^2(q)| &\leq |\sigma^2(p) - \sigma^2(\bar{p}_N)| + |\sigma^2(\bar{p}_N) - \sigma^2(q)| \\
&\leq 6\|f\|^2_{\mathsf{L}^\infty}\frac{\ell_p\sqrt{d}L^{d+1}}{2N} + 4 \|f\|_{\mathsf{L}^\infty} \ell_f\sqrt{d}\frac{L}{N} + 6\|f\|^2_{\infty}\epsilon\\
&:= \epsilon_\mathrm{var}.\label{eq:variance-bound}
\end{align}
where here we have used Eq.~\eqref{eq:variance_bound} from Lemma~\ref{lem:mean-variance-bound} as well as Eq.~\eqref{eq:var-tv-bound} from Lemma~\ref{lem:TV-consequences}.

Now, it follows immediately from Theorem~\ref{thm:discreteq}, via the fact that in this setting we have assumed access to the unitaries $U$ and $U^{-1}$, that the mean estimate $\tilde{\mu}$ output by the estimator after $O(m\log^{3/2}(m)\log\log (m))$ ``experiments'' satisfies
\begin{equation}\label{eq:first-bound}
\mathrm{Pr}\left[|\tilde{\mu}-\mu(q)|> \sigma(q)\frac{\log(1/\Delta)}{m}\right]\leq \Delta.
\end{equation}
In particular, when applying the mean estimator from Theorem~\ref{thm:discreteq} with access to $U$ and $U^{-1}$, one is essentially performing mean estimation with respect to $q$, and therefore gets a guaranteed estimator for $\mu(q)$. To be comparable to the classical estimator for $\mu$, we would now like to replace the dependencies on $\mu(q)$ and $\sigma(q)$ with dependencies on $\mu(p)$ and $\sigma(p)$. To do this, we first use the fact that $\sigma(q) \leq \sqrt{\sigma(p) + \epsilon_{\mathrm{var}}}$ (which follows from Eq.~\eqref{eq:variance-bound}), which together with Eq~\eqref{eq:first-bound} gives
\begin{equation}
\mathrm{Pr}\left(|\tilde{\mu}-\mu(q)|> \sqrt{\sigma^2(p) + \epsilon_{\mathrm{var}}}\frac{\log(1/\Delta)}{m} \right) \leq \mathrm{Pr}\left(|\tilde{\mu}-\mu(q)|> \sigma(q)\frac{\log(1/\Delta)}{m}\right)\leq   \Delta.
\end{equation}
Now, from the triangle inequality and Eq.~\eqref{eq:mean-triangle} we have 
\begin{equation}
|\tilde{\mu} - \mu(p)|\leq |\tilde{\mu} - \mu(q)| + |\mu(q) - \mu(p)| \leq |\tilde{\mu} - \mu(q)| + \epsilon_\mathrm{mean},
\end{equation}
and therefore
\begin{align}
\mathrm{Pr}\left(|\tilde{\mu}-\mu(p)|\leq \sqrt{\sigma^2(p) + \epsilon_{\mathrm{var}}}\frac{\log(1/\Delta)}{m} + \epsilon_{\mathrm{mean}}\right) \geq \mathrm{Pr}\left(|\tilde{\mu}-\mu(q)|\leq \sqrt{\sigma^2(p) + \epsilon_{\mathrm{var}}}\frac{\log(1/\Delta)}{m} \right) > 1-\Delta,
\end{align}
which proves the corollary.

\subsection{Proofs of Lemma~\ref{lem:mean-variance-bound} and Lemma~\ref{lem:TV-consequences}}
\label{ss:lemma_proof}

\begin{proof}[Proof of Lemma~\ref{lem:mean-variance-bound}]  For any $k = (k_1,\ldots, k_d)\in [N]^d$ define the cube-cell $C_k\subset [0,L]^d$ via
\begin{equation}
C_k = [k_1h,(k_1+1)h)\times \ldots \times [k_dh,(k_d+1)h).
\end{equation}
Note that $v_k = (k_1h,\ldots,k_dh)$ is the the lower-left vertex of $C_k$. We now define a helpful piecewise constant proxy $r$ for $p$ in two steps. Firstly, define $p^{(h)}$ via $p^{(h)}(x) = p(v_k)$ for all $x\in C_k$---i.e., $p^{(h)}$ is constant over each cell with the value of $p$ on the lower-left vertex of the cell. Next, normalize $p^{(h)}$ to be a valid PDF---i.e., define 
\begin{equation}
r(x) := \frac{p^{(h)}(x)}{Z}
\end{equation}
for all $x\in [0,L)^d = \cup_{k\in [N]^d}C_k$ where 
\begin{equation}
Z = \int_{[0,L)^d}p^{(h)}(x)dx = \sum_{k\in[N]^d}\int_{C_k} p(v_k) dx = h^d\sum_{k\in [N]^d}p(v_k) = h^d S.
\end{equation}
With this, we can proceed, and we start by proving Eq.~\eqref{eq:mean_bound}.

\textbf{Proof of Eq.~\eqref{eq:mean_bound}:} Via the triangle inequality we have
\begin{align}
|\mu(p) - \mu(\bar{p}_N)| \leq  \underbrace{|\mu(p) - \mu(r) |}_{\text{(A)}} + \underbrace{|\mu(r) - \mu(\bar{p}_N)|}_{\text{(B)}} \label{eq:error_decomposition}
\end{align}
Term (A) in the expression above is the difference in the mean of $f$ when approximating $p$ with its piecewise constant proxy, and term (B) is the difference when using the mean of $f$ over a cell (with respect to a constant probability over the cell), as opposed to simply the value of $f$ on the vertex. We now bound each of these two terms separately.

\textbf{Term (A):} Recalling the notation $\|f\|_{\mathsf{L}^\infty} = M$, we have that
\begin{equation}
\text{(A)} = \left|\int_{[0,1)^d} f(x)[p(x) - r(x)]dx\right| \leq M\|p - r\|_1.
\end{equation}
Now, using the fact that $r = p^{(h)}/Z$ one has
\begin{align}
\|p - r\|_1 &\leq \|p - p^{(h)}\|_1 + \|p^{(h)} - r\|_1 \\
&= \|p - p^{(h)}\|_1 + |Z-1|.
\end{align}
But, we also have 
\begin{align}
|Z-1| &= \left|\int p^{(h)}(x)dx - \int p(x)dx\right| \\
&=  \left|\int [p^{(h)}(x) - p(x)]dx\right|\\
&\leq \|p^{(h)}-p\|_1
\end{align}
and therefore
\begin{equation}
\|p - r\|_1 \leq 2\|p - p^{(h)}\|_1.
\end{equation}
We can now bound $\|p - p^{(h)}\|_1$ using the Lipschitz continuity of $p$. More specifically, for any $x\in C_k$ one has
\begin{align}
|p(x) - p^{(h)}(x)| &= |p(x) - p(v_k)| \\
&\leq \ell_p\|x-v_k\|_2 \\
&\leq \ell_p \mathrm{diam}(C_k)\\
&=\ell_ph\sqrt{d}.
\end{align}
Using this, we then have
\begin{align}
\|p - p^{(h)}\|_1 &= \int_{[0,1)^d}|p(x) - p^{(h)}(x)|dx \\
&= \sum_{k\in[N]^d}\int_{C_k}|p(x) - p^{(h)}(x)|dx \\
&\leq \sum_{k\in[N]^d}\int_{C_k}\ell_ph\sqrt{d}dx \\
&= \sum_{k\in[N]^d} \ell_ph\sqrt{d} \int_{C_k}dx \\
&= \sum_{k\in[N]^d} \ell_ph\sqrt{d}h^d\\
&= N^d\ell_ph\sqrt{d}h^d\\
&= \ell_p\sqrt{d}\frac{L^{d+1}}{N}.
\end{align}
As a result, we therefore have 
\begin{equation}\label{eq:1-norm-bound}
\|p-r\|_1 \leq 2 \ell_p\sqrt{d}\frac{L^{d+1}}{N}
\end{equation}
and 
\begin{equation}\label{eq:A_bound}
\text{(A)} \leq 2M \ell_p\sqrt{d}\frac{L^{d+1}}{N}.
\end{equation}
\textbf{Term (B):} As $r$ is constant on each cell $C_k$, we have
\begin{align}\label{eq:r_mean}
\mu(r) = \int_{[0,1)^d}f(x)r(x)dx = \frac{1}{Z}\sum_{k\in[N]^d}p(v_k)\int_{C_k}f(x)dx = \frac{h^d}{Z}\sum_{k\in[N]^d}p(v_k)\overline f(C_k),
\end{align}
where 
\begin{equation}
\overline{f}(C_k) = \frac{1}{h^d}\int_{C_k}f(x)dx
\end{equation}
is the average of $f$ on cell $C_k$. Additionally, by definition we have
\begin{equation}\label{eq:discrete_mean}
\mu(\bar{p}_N) = \sum_{k\in[N]^d}f(v_k)\frac{p(v_k)}{S} = \frac{h^d}{Z}\sum_{k\in[N]^d}f(v_k)p(v_k).
\end{equation}
Substituting Eqs.~\eqref{eq:r_mean} and~\eqref{eq:discrete_mean} into the expression for (B) then yields 
\begin{align}\label{eq:inter}
\text{(B)} &= \left| \frac{h^d}{Z}\sum_{k\in[N]^d}p(v_k)[\overline{f}(C_k) - f(v_k)]\right| \leq   \frac{h^d}{Z}\sum_{k\in[N]^d}p(v_k)\left|\overline{f}(C_k) - f(v_k)\right|.
\end{align}
Now, by Lipschitz continuity of $f$, and the fact that $\overline{f}(C_k)\leq \sup_{x\in C_k}f(x)$, we have
\begin{align}\label{eq:X_lip}
\left|\overline{f}(C_k) - f(v_k)\right| \leq \sup_{x\in C_k} |f(x) - f(v_k)| \leq \ell_f\mathrm{diam}(C_k) = \ell_fh\sqrt{d}.
\end{align}
Using this in Eq.~\eqref{eq:inter}, and recalling the definitions of $S$ and $Z$, we then have 
\begin{equation}\label{eq:B_bound}
\text{(B)} \leq \ell_f h\sqrt{d}.
\end{equation}
Substituting the bounds on terms (A) and (B) from Eqs.~\eqref{eq:A_bound} and~\eqref{eq:B_bound} respectively, into the error decomposition of Eq.~\eqref{eq:error_decomposition}, then proves Eq.~\eqref{eq:mean_bound}.
We now move onto the proof of Eq.~\eqref{eq:variance_bound}.

\textbf{Proof of Eq.~\eqref{eq:variance_bound}:} We again start with an error decomposition via the triangle inequality as follows
\begin{align}
|\sigma^2(p) -\sigma^2(\bar{p}_N)| &= |(m_2(p) - [\mu(p)]^2) - (m_2(\bar{p}_N) - [\mu(\bar{p}_N)]^2) | \\
&\leq \underbrace{|m_2(p)- m_2(r)|}_{\text{(C)}} +  \underbrace{|m_2(r)- m_2(\bar{p}_N)|}_{\text{(D)}} +  \underbrace{|[\mu(p)]^2- [\mu(r)]^2|}_{\text{(E)}} + \underbrace{|[\mu(r)]^2- [\mu(\bar{p})]^2|}_{\text{(F)}}.\label{eq:triangle_decomp}
\end{align}
We also again proceed term by term.

\textbf{Term (C)}: Similarly to Term (A), here we have
\begin{align}
\text{(C)} = \left|\int_{[0,L)^d} f^{2}(x)[p(x) - r(x)]dx\right| \leq M^2\|p-r\|_1 
\leq 2M^2\ell_p\sqrt{d}\frac{L^{d+1}}{N},
\end{align}
where in the last step we have used Eq.~\eqref{eq:1-norm-bound}.

\textbf{Term (E)}: Here we have
\begin{align}
\text{(E)} &= |[\mu(p)]^2- [\mu(r)]^2| \\
&= |\mu(p)- \mu(r)||\mu(p) +  \mu(r)|\\
&\leq |\mu(p)- \mu(r)|\left(|\mu(p)| + |\mu(r)|\right) \label{eq:first_step} \\
&\leq 2M|\mu(p)- \mu(r)\label{eq:second_step}|\\
&\leq 4M^2\ell_p\sqrt{d}\frac{L^{d+1}}{N}\label{eq:third_step}
\end{align}
where in going from line~\eqref{eq:first_step} to~\eqref{eq:second_step} we have used that $\mu(Q) \leq \|f\|_{\mathsf{L}^\infty} = M$ for any distribution $Q$, and in going from line~\eqref{eq:second_step} to~\eqref{eq:third_step} we have used the bound for Term (A) from Eq.~\eqref{eq:A_bound}.

\textbf{Term (D):} Similarly to term (B) we have
\begin{align}
& m_2(r) = \frac{h^d}{Z}\sum_{k\in[N]^d}p(v_k)\overline{f}^2(C_k) \\
& m_2(\bar{p}_N) = \frac{h^d}{Z}\sum_{k\in[N]^d}p(v_k)f^2(v_k),
\end{align}
where 
\begin{equation}
\overline{f}^2(C_k) = \frac{1}{h^d}\int_{C_k}f^2(x)dx.
\end{equation}
We therefore have that 
\begin{equation}\label{eq:d-simplification}
\text{(D)} \leq \frac{h^d}{Z}\sum_{k\in[N]^d}p(v_k)\left|\overline{f}^2(C_k)-f^2(v_k)\right|.
\end{equation}
Now, by Lipschitz continuity of $f$ we have
\begin{align}
\left|\overline{f}^2(C_k)-f^2(v_k)\right| &\leq \sup_{x\in C_k} \left|f^2(x) - f^2(v_k)\right|, \\
&=\sup_{x\in C_k} \left|[f(x) - f(v_k)][f(x) + f(v_k)]\right|,\\
&\leq 2M \sup_{x\in C_k} \left|f(x) - f(v_k)\right|,\\
&\leq 2M\ell_fh\sqrt{d},\label{eq:final}
\end{align}
where in the last line we have used Eq.~\eqref{eq:X_lip}. Now, substituting~\eqref{eq:final} into~\eqref{eq:d-simplification} and recalling the definitions of $S$ and $Z$ yields
\begin{equation}
\text{(D)}\leq 2M\ell_fh\sqrt{d}.
\end{equation}
\textbf{Term (F):} Similarly to term (E) we have
\begin{align}
\text{(F)} &= |[\mu(r)]^2- [\mu(\bar{p}_N)]^2| \\
&=  |\mu(r)- \mu(\bar{p})||\mu(r)]^2+ \mu(\bar{p})| \\
&\leq 2M |\mu(r)- \mu(\bar{p}_N)| \\
&\leq 2M \ell_fh\sqrt{d},
\end{align}
where the last line follows from noting that  $|\mu(r)- \mu(\bar{p}_N)|$ is precisely term (B) which has already been bounded in Eq.~\eqref{eq:B_bound}.

Finally, substituting the bounds for terms (A),(B),(C) and (D) into Eq.~\eqref{eq:triangle_decomp} gives the proof of Eq.~\eqref{eq:variance_bound}.
\end{proof}

\begin{proof}[Proof of Lemma~\ref{lem:TV-consequences}] We start by proving Eq.~\eqref{eq:mean-tv-bound} as follows:
\begin{align}
|\mu(P) - \mu(Q)| &= \left|\sum_{x}f(x)[P(x) - Q(x)]\right| \\
&\leq \|f\|_{\mathsf{L}^\infty} \sum_{x}|P(x) - Q(x)|\\
&\leq 2\|f\|_{\infty}\epsilon,\label{eq:TV_consequence}
\end{align}
exploiting the definition of the TV distance and the fact that $\mathrm{d}_\mathrm{TV}(\bar{p}_N,q)\leq \epsilon$ in the last line. Similarly Eq.~\eqref{eq:var-tv-bound} is obtained via
\begin{align}
|\sigma^2(P) - \sigma^2(Q)| &= \left| m_2(P) - [\mu(P)]^2 - (m_2(Q) - [\mu(Q)]^2)\right| \\
&\leq |m_2(P) - m_2(Q)| + \left|[\mu(P)]^2 - [\mu(Q]^2\right| \\
&\leq \left|\sum_{x} f^2(x)[P(x)-Q(x)]\right| + |\mu(P) + \mu(Q)||\mu(P) - \mu(Q)|\\
&\leq \|f\|^2_{\mathsf{L}^\infty}\sum_{x} |P(x)-Q(x)| + \left(|\mu(P)| + |\mu(Q)|\right)|\mu(P) - \mu(Q)|\\
&\leq 2\|f\|^2_{\mathsf{L}^\infty}\epsilon + 2\|f\|_{\mathsf{L}^\infty}|\mu(P) - \mu(Q)|\\
&\leq 2\|f\|^2_{\mathsf{L}^\infty}\epsilon + 4\|f\|^2_{\mathsf{L}^\infty}\epsilon\\
&= 6\|f\|^2_{\mathsf{L}^\infty}\epsilon.
\end{align}
\end{proof}

\end{document}